\documentclass[]{article}

\usepackage[english]{babel}
\usepackage{amsmath, amsthm, amsfonts, amssymb, accents}
\usepackage{graphicx, dsfont, srcltx, color}

\numberwithin{equation}{section}
\theoremstyle{plain}

\newtheorem{lemma}{Lemma}
\newtheorem{remark}{Remark}

\def\Tht{\Theta}

\def\e{\varepsilon}
\def\g{\gamma}

\def\l{\lambda}
\def\p{\partial}

\def\a{\alpha}
\def\b{\beta}

\def\Op{\mathcal{H}}

\newcommand{\PT}{\mathcal{PT}}

\renewcommand{\geq}{\geqslant}

\renewcommand{\leq}{\leqslant}

\DeclareMathOperator{\RE}{Re}
\DeclareMathOperator{\IM}{Im}
\DeclareMathOperator{\sign}{sign}
\DeclareMathOperator{\arcsinh}{arcsinh}

\title{Spacing gain and absorption in a simple $\PT$-symmetric model: spectral singularities and  ladders of eigenvalues and resonances}

\author{D. I. Borisov$^{1,2,3}$\\
$^1$Institute of Mathematics, Ufa Federal Research Center,
\\
Russian Academy of Sciences,
Ufa, Russia, 450008
\\[3mm]
$^2${Bashkir State Pedagogical University named after M.Akhmulla}, 
Ufa, Russia, 450000
\\[3mm]
$^3$University of Hradec Kr\'alov\'e,
Hradec Kr\'alov\'e 50003, Czech Republic
\\[5mm]	
D. A. Zezyulin$^4$\\
$^4$ITMO University, St. Petersburg 197101, Russia}

\allowdisplaybreaks

\begin{document}

\maketitle

\begin{abstract}
	
	We consider a parity-time ($\PT$-) symmetric waveguide   consisting of a localized gain and loss elements separated by a variable distance. The situation is modelled by a Schr\"odiner operator with localized complex $\PT$-symmetric potential. Properties of the latter Hamiltonian are considered subject to the change of the gain-to-loss distance. Resonances, spectral singularities and eigenvalues  	 are analyzed in detail and discussed in the context of the associated laser-absorber modes and $\PT$-symmetry breaking phase transition. Increasing gain-to-loss distance  creates new resonances and spectral singularities which do not exist in the waveguide with adjacent gain and loss. In the limit of large gain-to-loss  distance, the waveguide features a ladder of resonances  which can be transformed to a ladder of complex eigenvalues by means of the change of the gain-and-loss amplitude.
		
\end{abstract}

\section{Introduction}

Waves interacting with localized parity-time ($\PT$-) symmetric potentials feature a variety  of novel and interesting phenomena, such as  scattering resonances at real energies \cite{Mostafazadeh2009},  coherent perfect absorption and lasing  operation \cite{Longhi10}, loss-induced lasing \cite{Nori}, invisibility effects \cite{invisib} and others. The simplest realization of a $\PT$-symmetric system couples a pair of  geometrically identical elements with energy gain and absorption \cite{pendula}. This idea was, in particular, elaborated in the context of electromagnetic waveguides \cite{Mostafazadeh2009,Muga05},  coupled optical structures \cite{El-Ganainy07,Ruter10}, and in many other physical environments, see reviews \cite{PukhovReview,KZYreview,LonghiReview,FengReview,El-GanainyReview}.  Behavior of many bicentric $\PT$-symmetric structures can be modelled by effective Schr\"odinger-like Hamiltonians,   which are usually non-self-adjoint (with respect to the standard inner product).  In the present paper, we consider a one-dimensional Hamiltonian with a generalized  $\PT$-symmetric  step potential    consisting of a pair of identical gain and loss elements separated by a variable distance.  We aim to study the  spectral properties of this operator as the distance varies. In our study, we mostly focus on two aspects. The first direction concerns the   spectral properties of our operator as the distance between gain and loss elements is large. We demonstrate that  in this regime our  simple bicentric $\PT$-symmetric Hamiltonian  features a ladder of resonances  resembling the well-known Wannier-Stark ladder  for periodic Hamiltonians with an additional potential gradient \cite{SW6}. Upon   the change of the  gain-and-loss amplitude the ladder of   resonances can be transformed  to a ladder of complex-conjugate pairs of eigenvalues and vice versa. This  behavior has no counterpart  in the self-adjoint case. Transformation from a resonance to  an eigenvalue corresponds to a zero-width resonance, i.e. to a spectral singularity in the continuous spectrum \cite{Mostafazadeh2009}. A detailed study of spectral singularities and associated $\PT$-phase transition under   the change of the  gain-to-loss distance  is the second objective of our work.

To the best of the authors' knowledge, $\PT$-symmetric problems with gain and losses separated by a large distance, were not studied before. At the same time, this regime was studied in details for self-adjoint  Schr\"odginer operators with two (or several) potential wells separated by a large distance.
Not trying to cite all papers in this direction, we mention   several   classical works
\cite{D, AS, H, MS, KS1, GHS} and a series of very recent ones \cite{UMJ,BEG, GAM, GAM2}. The classical works addressed the case of several wells, while in the recent works the wells were replaced by
localized abstract operators, including such examples as second order operators with compactly supported coefficients, localized geometric perturbation, delta-potentials on compact manifolds, etc. The main   results concerned the behavior as   the distances between localized wells or operators tends to infinity. In such a regime, the resolvent of the original operators is approximated in the norm resolvent sense by the direct sum of the resolvents for each well or localized operator. The most general results of such kind were obtained in \cite{UMJ, GAM}. The eigenvalues and the associated eigenfunctions in such regime also attracted a lot of interest. In the classical works, the leading terms of their asymptotics were obtained. A general scheme allowing to construct complete asymptotic expansions in the most general setting was proposed in \cite{GAM2}.  We also mention work \cite{BEG}, in which the localized operators were modeled by change of type of boundary conditions and there were studied resonances emerging from the eigenvalues embedded into the essential spectrum.

As we have mentioned above, in the large distance regime we find a phenomenon similar to Wannier-Stark resonances; let us briefly recall the matter of such resonances. Wannier-Stark resonances or Wannier-Stark ladders is a
celebrated topic studied both by physicists  and mathematicians, see, for instance, papers \cite{SW6, SW1, SW2, SW3, SW4, SW5}.
Here the main object the study is a   Schr\"odinger operator with a fixed periodic potential and a linear potential with a small coupling constant like
\begin{equation}\label{0.1}
-\frac{d^2\ }{dx^2}+V_{per}(x)+\e x\quad\text{on}\quad\mathds{R}.
\end{equation}
The main phenomenon is that as $\e\to0$, such operator has a growing number of resonances accumulating to some zone in the essential spectrum. Namely, the imaginary parts of these resonances and the distances between two neighbouring resonances tend to zero, that is, there is a ladder of resonances in the vicinity of some part of the essential spectrum. A distinctive feature of this phenomenon is that the mentioned zone in the essential spectrum {\sl contains
no spectral singularities} of the considered operator as $\e=0$.
 In fact, the ladder of mentioned resonances can be regarded as emerging from the eigenvalues of the operator $-\frac{d^2\ }{dx^2}+\e x$ accumulating on the real axis as $\e\to+0$.  The main obtained results describe the total number and the asymptotic behavior of these resonances. Wannier-Stark ladders and  the associated Bloch oscillations have been addressed  theoretically for a $\PT$-symmetric version of operator (\ref{0.1}) in \cite{SW-PT-1}
and realized experimentally in  synthetic $\PT$-symmetric photonic lattices \cite{SW-PT-2, SW-PT-3}.

In contrast to  the conventional, Hermitian operators,  in $\PT$-symmetric systems, scattering resonances have more chances to occur at real energies, i.e. at   interior points of the continuous spectrum \cite{Mostafazadeh2009,ScarfII}.  A real-energy resonance has zero width  and corresponds to a spectral singularity \cite{SS1,SS2,SS4,SS5}, i.e. to a real pole of the resolvent. Additional source of steady interest in spectral singularities originates in the associated regime of    laser--antilaser operation \cite{Longhi10,SS5,Stone_selfdual}:  at the spectral singularity  the system admits a pair of linearly independent solutions,  one of which has a form of a coherent  bidirectional  outgoing radiation   (``laser''), and another solution is a coherently absorbed   incident radiation (``antilaser'' \cite{Rozanov} or coherent perfect absorber (CPA) \cite{Stone1,Stone2,reviewCPA}). Under a small perturbation of the $\PT$-symmetric  operator, the spectral singularity can split either in a pair of resonances or in  a complex-conjugate pair of discrete eigenvalues bifurcating from an interior point of the continuous spectrum  \cite{Garmon,Yang17,KZ17,Ahmed18,Konotop2019,BorDmi}. The latter situation is related to a peculiar scenario of an exceptional-point-free $\PT$-symmetry phase transition from purely real to complex spectrum   \cite{Yang17,KZ17,Konotop2019},  which is distinctively different from the better studied $\PT$-phase transition through an exceptional point \cite{Kato,Heiss}. The detailed study of spectral singularities and the associated phase transition in the simple $\PT$-symmetric system with a variable gain-to-loss distance is the second goal of the present paper.

In the present paper
we consider a one-dimensional $\PT$-symmetric Schr\"odinger operator with a complex potential formed by two localized parts describing losses and gain and separated by   some distance. The second parameter in the model is the gain-and-loss amplitude. We study the resonances and the eigenvalues of this model. We derive a compact equation describing the location of the corresponding finite- and zero-width resonances as well as  the emergence of localized bound states associated with discrete complex-conjugate eigenvalues.
Then we prove a series of qualitative properties of these objects, namely, that there are countably many resonances accumulating at infinity and only finitely many eigenvalues and all of them are symmetric with respect to the imaginary axis. We provide explicit description of the domains in the complex plane, in which the presence of the eigenvalues and resonances is ensured.

The first main result concerns the asymptotic limit of infinitely large distance between the gain and \textcolor{black}{loss}. It turns out that in such regime, there exists a ladder of resonances or eigenvalues accumulating in the vicinity of a certain zone of the essential spectrum. This zone is an explicitly found segment starting at zero. We show that as the distance exceeds a certain  explicitly given threshold,  there are  $N$ eigenvalues or resonances in the vicinity of the mentioned segment, where $N$ is also given explicitly in terms of the distance. We also find an asymptotic four-terms approximations for these eigenvalues and resonances with explicitly written terms. These approximations show that the loss-and-gain amplitude determines what kind of objects, eigenvalues or resonances, we deal with, while the distance is responsible for the number of these objects and their distribution. As the distance grows, these eigenvalues or resonances go \textcolor{black}{closer  to each  other} and also tend to the real axis, and this is exactly the same picture as for the Wannier-Stark resonances. At the same time, as we discuss in Section~\ref{sec:zeros}, the nature of this ladder is different from Wannier-Stark's case.  There are several advantages of our model. The first of them is that such picture is generated not by a linear potential multiplied by a small parameter, but simply by spacing two compactly supported potentials.  Thanks to the $\PT$-symmetry, we can generate not only resonances, but also eigenvalues(!), which is a completely new phenomenon indicating once again how rich the spectra of $\PT$-symmetric operators are. In comparison to a rather complicated for analysis operator (\ref{0.1}), our model is much simpler, and this is why we are able to find explicitly the most important characteristics like the total number of zeroes, the threshold for the distance ensuring such phenomenon and several terms in the asymptotic expansions for eigenvalues or resonances.

We also study zero-width resonances of our model. Using a combination of analytical and numerical approaches, we show  that the increase of the gain-to-loss distance can  drive the  $\PT$-symmetric waveguide through a sequence of self-dual spectral singularities each of  which   corresponds to the laser-antilaser regime and leads either to the emergence of a  new pair of complex-conjugate eigenvalues or to the immersion of the complex-conjugate pair in the continuum. Practical importance of our findings consists in the possibility to use the distance between the gain and loss as a parameter allowing   to decrease the $\PT$-symmetry breaking laser-antilaser threshold: this means that in a waveguide with mutually spaced gain and losses the laser-antilaser regime can be achieved at a smaller gain-and-loss strength than in a waveguide with the adjacent gain and loss elements. The distance between the gain and loss elements enables an additional degree of freedom which allows to achieve a spectral singularity at any real wavenuber given beforehand and also to achieve two spectral singularities simultaneously.

The content of this  paper is organized as follows. In Section~\ref{sec:model} we introduce the  $\PT$-symmetric model and the basic terminology, and derive the main equation determining the location of resonances and eigenvalues in the complex plane. In Section~\ref{sec:lemmata} we prove a series of lemmata describing  the properties of complex roots. In Section~\ref{sec:zeros} we apply  the proven statements in order to characterize the general picture of the location and behavior of the zeroes in the complex plane and  we discuss the  ladders of eigenvalues and resonances accumulating to the essential spectrum.
Section~\ref{sec:ss} is devoted to the detailed study of the spectral singularities of $\PT$-phase transition. Section~\ref{sec:concl} concludes the paper.

\section{$\PT$-symmetric waveguide with spaced gain and loss}
\label{sec:model}

\subsection{General model and   periodicity in the gain-to-loss distance}
\label{sec:periodicity}

We consider stationary scattering on a spatially localized $\PT$-symmetric potential $V(x)$:
\begin{equation}\label{2.1}
-\psi_{xx} + V(x)\psi = \l    \psi,\qquad \l=k^2, \qquad x\in\mathds{R},
\end{equation}
where $\l$ is the eigenvalue and $k$ is the complex wavenumber. The complex potential  $V(x)$ describes the internal structure of the waveguide.  We consider a  pair of geometrically identical elements with gain and loss separated by some distance. Since in this work we are primarily interested in the role  of  the gain-to-loss distance, we assume   that the   width of each element is scaled to unity. In order to model this situation, we introduce a complex-valued function $W(x)$ compactly supported in the interval $(0, 1)$, and    consider a localized complex   potential in the form  (see figure~\ref{fig:potential})
\begin{equation}\label{2.2}
V(x) = \left\{
\begin{aligned}
&W(x+\ell+1)\qquad &&\text{for}\quad x\in (-\ell-1,-\ell),
\\
&\overline{W}(-x+\ell+1)\qquad &&\text{for}\quad x\in (\ell,\ell+1),
\\
&\hphantom{W(x+}0 &&\text{for all other $x$},
\end{aligned}\right.
\end{equation}
where bar is the complex conjugation, and $\ell\geqslant 0$  is the half-distance between the two elements (in order to shorten  the discussion, in what follows we will use term ``gain-to-loss distance'' for $\ell$); the case $\ell=0$ corresponds to the adjacent elements.  Potential $V(x)$ satisfies the standard condition of $\PT$ symmetry: $V(x)=\overline{V}(-x)$. In figure~\ref{fig:potential} we show a   schematics of the potential consisting of an active  (``gain'') and absorptive (``losses'')   elements.

\begin{figure}
	\centering
	\includegraphics[width=0.5\columnwidth]{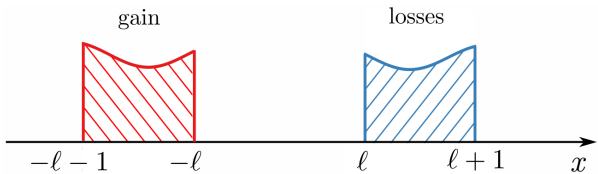}%
	\caption{Schematics of a localized potential which consists of a pair of geometrically identical active (gain) and absorptive (losses)   elements which are   of  unitary width  and   separated by the distance $2\ell$.}
	\label{fig:potential}
\end{figure}

Spectral singularities, finite-width resonances, and  eigenvalues are associated with non-trivial solutions $\psi(x)$ of equation (\ref{2.1}) with the following behavior at infinity:
\begin{equation}\label{2.3}
 \psi(x)=C_- e^{ikx}, \quad \mbox{for\ } x<-\ell-1, \quad \mbox{and\quad } \psi(x)=C_+e^{-ikx}, \quad  \mbox{for\ } x>\ell+1,
\end{equation}
where $C_+$ and $C_-$ are  constants. Values of $k$ in the upper complex half-plane, Im$\,k>0$, are associated with  exponentially growing   generalized eigenfunctions which in the scattering theory correspond to   finite-width resonances. The case Im$\,k=0$ corresponds to   spectral singularities, i.e. to the zero-width resonances inside the contionuous spectrum, i.e. at   $\l=k^2\geq 0$ \cite{Mostafazadeh2009}. Positive and negative  $k$ correspond to the coherent perfect absorber (antilaser) and laser solutions, respectively. 
Values of $k$ in the lower complex half-plane, Im$\,k<0$, correspond to discrete eigenvalues associated with spatially localized bound states.

Equations (\ref{2.1})--(\ref{2.2}) with condition (\ref{2.3}) can be solved as
\begin{align*}
&
\psi(x,\ell)=e^{i k (x+\ell+1)}, && x<-\ell-1,
&& \psi(x,\ell)=\psi_-(x+\ell), && -\ell-1<x<-\ell,
\\
&
\psi(x,\ell)=C e^{-i k (x-\ell-1)}, && x>\ell+1,
&&
\psi(x,\ell)=C\psi_+(x-\ell),
&&\hphantom{1,.} \ell<x<\ell+1,
\end{align*}
where $C$ is some constant and functions $\psi_-$ and $\psi_+$ describe  the solution inside the active and absorptive cells; their specific shapes are determined by function $W$ and are independent of $\ell$.  The requirement of the continuity of the wave function  $\psi(x)$ and of its derivative implies that $\psi_-(-1)=1$, $\psi_-'(-1)=ik$, $\psi_+(1)=1$, and $\psi_+'(1)=-ik$.

In  the central region, the requirement of continuity of $\psi$ and $\psi'$ implies that for $|x|<\ell$ the solution has the form
\begin{equation}\label{2.8}
\psi(x)=\psi_-(0)\cos k (x+\ell)+ \psi_-'(0)\frac{\sin k(x+\ell)}{k}.
\end{equation}
At the same time,  the solution  for $|x|<\ell$ can be also represented as
\begin{equation}\label{2.9}
\psi(x)=C\psi_+(0)\cos k (x-\ell)+ C\psi_+'(0)\frac{\sin k(x-\ell)}{k}.
\end{equation}
Equations (\ref{2.8}) and (\ref{2.9}) are compatible   only if $k$ solves  the equation
\begin{equation}
\label{eq:main}
\big(\psi_-'(0)+i k \psi_-(0)\big) \big(\psi_+'(0)-i k \psi_+(0)\big)=
e^{-4i k\ell} \big(\psi_+'(0)+i k \psi_+(0)\big) \big(\psi_-'(0)-i k \psi_-(0)\big).
\end{equation}
The values    $\psi_-(0)$,  $\psi'_-(0)$, $\psi_+(0)$,  $\psi'_+(0)$  are independent of  $\ell$. This means that the gain-to-loss distance $\ell$ enters the main equation (\ref{eq:main}) only in the exponential term $e^{-4i k\ell}$, which is periodic in  $\ell$.  This  observation which will be used in what follows.

\subsection{Generalized $\PT$-symmetric step function}

The simplest implementation of potential (\ref{2.2}), which will be considered in the rest of this paper, corresponds to the case when  $W(x)$  is a rectangular function with a purely imaginary amplitude, i.e. $W(x) = i\gamma$ for $x\in(0,1)$ and $W(x)=0$ outside this interval. Then   $\g>0$ is the gain-and-loss amplitude,  which determines the gain strength for the left rectangle and the absorption rate for the right rectangle. The resulting potential $V(x)$ is a generalization of the  $\PT$-symmetric step function with adjacent gain and loss regions, which corresponds to $\ell=0$ and was considered earlier in   \cite{Mostafazadeh2009,Muga05,KZ17,Ahmed18,Ahmed04}. In the present paper,   we are   concerned with the effect of the nonzero gain-to-losses distance $\ell>0$.  The obtained generalized   $\PT$-symmetric step potential   can be of relevance in various physical settings. First, it arises in the study of the scattering of TE electromagnetic waves in a rectangular waveguide containing gain and losses   such that  the complex-valued dielectric permittivity changes along the direction of the propagation.  This can be implemented by filling regions of  the waveguide with resonant atoms \cite{Mostafazadeh2009,Muga05} (the gain and loss media are separated by vacuum in order to create nonzero distance $\ell$).  Another well-established setting, where $\PT$-symmetric potential (\ref{2.2}) can arise,  corresponds to the propagation of a paraxial optical beam in a system of coupled waveguides with gain and losses \cite{El-Ganainy07,Ruter10}.  In this case the variation of complex-valued dielectric permittivity takes place in the transverse direction, while the beam propagates in the longitudinal direction. Similar bicentric $\PT$-symmetric potentials (but in the idealized double-delta  function form) have been also considered in the theory of   Bose-Einstein condensates, where the domains with gain and losses correspond to the pump and absorption of particles, respectively \cite{Cartarius12,Zezyulin16}.

For the generalized $\PT$-symmetric step potential, functions $\psi_-$ and $\psi_+$ can be evaluated easily:
\begin{align}
&\psi_-(x)=\cos\sqrt{k^2-i\g}(x+1)
+i k \frac{\sin\sqrt{k^2-i\g}(x+1)}{\sqrt{k^2-i\g}}, &&
-1<x<0,
\\
&\psi_+(x)= \cos\sqrt{k^2+i\g} (x-1)
-i k \frac{ \sin\sqrt{k^2+i\g}(x-1)}{\sqrt{k^2+i\g}}, && \hphantom{.11}
0<x<1.
\end{align}
Using these formulae to compute   $\psi_-(0)$,  $\psi'_-(0)$, $\psi_+(0)$,  $\psi'_+(0)$  and substituting the latter into equation (\ref{eq:main}),   we arrive at  a final equation for $k$:
\begin{align}\label{2.10}
&F(k,\ell,\g)=0,\qquad F(k,\ell,\g):=F_-(k,\g)F_+(k,\g)-e^{-4ik\ell} F_0(k,\g),
\\
&
F_\pm(k,\g):=2ik\cos\sqrt{k^2\pm i\g} -(2k^2\pm i\g)\frac{\sin\sqrt{k^2\pm i\g}}{\sqrt{k^2\pm i\g}},
\quad
F_0(k,\g):= \g^2   \frac{\sin\sqrt{k^2-i\g}}{\sqrt{k^2-i\g}} \, \frac{\sin\sqrt{k^2+i\g}}{\sqrt{k^2+i\g}}.\nonumber
\end{align}

Our next step is a   detailed   study  of  the zeroes of   function $F$.

\section{Auxiliary lemmata}
\label{sec:lemmata}

In this section we point out  some obvious properties of the zeroes of the function $F$ defined in (\ref{2.10}) and prove several auxiliary lemmata. These will be used in the next section to describe a general picture of the zeroes location and their behavior as $\ell$ and $\g$ vary.

We begin with  obvious properties. First, it is straightforward to check that $k=0$ is zero for all $\ell$ and $\g$, that is, $F(0,\ell,\g)=0$, $\ell\geqslant 0$, $\g\geqslant 0$. Next, we notice that if $k$ is a zero, then $-\overline{k}$ is also a zero, that is, the zeroes of $F$ are symmetric with respect to the imaginary axis. The associated non-trivial solution of equation (\ref{2.1}) are related by the $\PT$ symmetry transform, i.e. by the combination of spatial reversal $x\mapsto -x$ and the complex conjugation. In particular, for any real root $k>0$ there always  exists  the opposite root $-k<0$, which means that the corresponding spectral singularity is self-dual and   associated with   the combined laser-antilaser regime.

If $k$ is a zero with a negative imaginary part, then the imaginary part of the corresponding eigenvalue $\l=k^2$ does not exceeds $\g$. Indeed, if $\psi$ is an associated eigenfunction, then due to the equation (\ref{2.1}) we have
\begin{equation*}
\|\psi'\|_{L_2(\mathds{R})}^2+(V\psi,\psi)_{L_2(\mathds{R})}= \l  \|\psi\|_{L_2(\mathds{R})}^2,
\end{equation*}
and the inequality $|\IM\l|\leqslant \g$  is obtained immediately by taking the imaginary part of this identity. The latter property can be equivalently reformulated as follows: if $k$ is a zero with a negative imaginary part, then
\begin{equation}\label{3.0b}
|\RE k\,\IM k|\leqslant \frac{\g}{2}.
\end{equation}

As was already pointed out in Section~\ref{sec:periodicity}, the gain-to-loss distance $\ell$ enters   function $F$ only in the exponent $e^{-4ik\ell}$. This means that if some zero $(k_0, \ell_0, \g_0)$ of function $F$ is found with real $k_0$, then one can immediately generate a sequence of distances
\begin{equation}
\label{eq:periodic}
\ell_n = \ell_0 + \frac{n\pi}{2k_0},
\end{equation}
which also correspond to zeros of function $F$ with the same $k_0$ and $\g_0$. In equation (\ref{eq:periodic}), $n$ is any integer such that $\ell_n\geqslant 0$.

We proceed to auxiliary lemmata. The first lemma describes the existence and the number of the zeroes.

\begin{lemma}\label{lm3.1}
For each $\g\geqslant 0$, $\ell\geqslant 0$, the function $F(\cdot,\ell,\g)$ is holomorphic in the entire complex plane, i.e. it is entire. For each $\g>0$, $\ell\geqslant 0$, this function has countably many complex zeroes accumulating at infinity only. These zeroes are continuous in $\ell$ and $\g$.
\end{lemma}

\begin{proof}
We represent the sine and cosine functions in the definitions of $F_\pm$, $F_0$ by their Taylor series. Then we see that the final expansion does not involve square roots but only the integer powers of $k$. Hence, the function $F(\cdot,\ell,\g)$ is entire.

The function $F(k,\ell,\g)$ satisfies the estimate $|F(k,\g)|\leqslant e^{C|k|}$ for some fixed constant $C$ and therefore, this is an entire function of order $1$. It is clear that as $\g>0$, for each constant $a\in\mathds{C}$, the function $e^{-a k}F(k,\g)$ is not a polynomial in $k$. Then, by the Hadamard theorem \cite[Sect. 4.2, Thm. 4.1]{Le}, the function $F(\cdot,\g)$ has countably many zeroes accumulating at infinity only. The continuity of these zeroes on $\ell$ and $\g$ is thanks to the smoothness of $F$.
The proof is complete.
\end{proof}

\begin{remark}\label{lm2}
As we have showed in the above proof,  the choice of the square root in the definition of $F_\pm$ is not important. The only restriction is that in the quotients $\frac{\sin\sqrt{k^2\pm i\g}}{\sqrt{k^2\pm i\g}}$ the branch should  be the same.
\end{remark}

\begin{remark}\label{rm1}
As $\g=0$, the function $F$ becomes very simple: $F(k,\ell,0)=-4k^2e^{2ik}$ and it has the only zero $k=0$ of second order. The associated nontrivial solution of equation (\ref{2.1}) is constant: $\psi(x)\equiv 1$.
\end{remark}

The second lemma is devoted to the case of small $\g$.

\begin{lemma}\label{lm3.2}
As $\g$ is small enough, for each $\ell\geqslant 0$, there are exactly two  zeroes of $F$ satisfying $|k|\leqslant 2\sqrt{\g}$.  One of these zeroes is $k=0$ for all $\ell$ and $\g$, while the other is pure imaginary and depends holomorphically \textcolor{black}{on} $\g$ and $\ell$. The leading terms of its Taylor series are as follows:
\begin{equation}\label{3.1}
k(\g,\ell)=i\left(\ell+\frac{1}{3}\right)\g + i\left(
2\ell^3+\frac{8}{3}\ell^2+\frac{53}{45}\ell+\frac{53}{315}\right)\g^4+O(\g^6).
\end{equation}
\end{lemma}

\begin{proof}
According to Remark~\ref{rm1}, for $\g=0$, the function $F$ has the only zero of second order. Since this function depends continuously on $\g$, by the Rouch\'e theorem we immediately infer that the function $F$ has exactly two zeroes (counting orders) converging to $k=0$ as $\g=0$.
It is easy to confirm that $F(0,\ell,\g)=0$ for all $\ell$ and $\g$ and hence, $k=0$ is one of the mentioned zeroes. The other zero is  pure imaginary since all zeroes are symmetric with respect to the imaginary axis.

The function $G(k,\ell,\g):=k^{-1}F(k,\ell,\g)$ is jointly holomorphic in $k$, $\ell$, $\g$ and by straightforward calculations we confirm that
\begin{equation}\label{3.0}
G(k,\ell,\g)=4i\left(\ell+\frac{1}{3}\right)\g^2-4k+O(\g^3+k\g),\qquad
G(0,\ell,0)=0,\qquad \frac{\p G}{\p k}(0,\ell,0)=-4\ne0,
\end{equation}
Then by the inverse function theorem we conclude that there exists the unique zero $k(\ell,\g)$  of $G$ in the circle $|k|\leqslant 2\sqrt{\g}$. This zero converges to  $k=0$ as $\g\to0$ and  is jointly holomorphic in $\ell$ and $\g$ for sufficiently small $\g$. As (\ref{3.0}) suggests, this zero is of order $O(\g^2)$ as $\g\to0$. We write the leading terms of its Taylor series
\begin{equation}\label{3.5}
k(\ell,\g)=c_2\g^2+c_3\g^3+c_4\g^4+c_5\g^5+O(\g^6),
\end{equation}
substitute this expression into the equation $G(k,\ell,\g)=0$, expand the result into power series in $\g$ and equate the coefficients at the like powers of $\e$. This  determines the coefficients $c_j$ in (\ref{3.5}) and a final form of this formula is exactly (\ref{3.1}).  The proof is complete.
\end{proof}

The next lemma provides a rough information on the location of the zeroes of the function $F(k,\ell,\g)$.

\begin{lemma}\label{lm3.3}
For each $\g>0$, $\ell\geqslant 0$, all zeroes of $F(\cdot,\ell,\g)$ are located in the domain
\begin{equation*}
\left\{k:\, |k|>\sqrt{\g}r,\ \IM k>0,\
\sqrt{\g}e^{(\ell+1)\IM k}<|k|<\frac{47}{25}
 \sqrt{\g}e^{(\ell+1)\IM k},\ |k|<\frac{76}{73}\g e^{(\ell+\frac{9}{8})|k|}
 \right\}\cup\Big\{k:\, |k|\leqslant \sqrt{\g}r
\Big\},
\end{equation*}
where
\begin{equation}\label{3.5a}
r:=\max\left\{\frac{39}{20},\frac{\sqrt{\g}}{2}\right\}.
\end{equation}
If $\g$ is small enough,   the function $F$ has no zeroes in the lower complex half-plane.
\end{lemma}

\begin{proof}
Let $k$ be a zero of $F$ and denote $z:=k\mu^{-1}$, $\mu:=\sqrt{\g}$.
By routine straightforward calculations we confirm that  $z$ solves the equation
\begin{align}
 &Q(z,\ell,\mu)=0,\qquad Q(z,\ell,\mu):=Q_1(z,\mu)-Q_2(z,\ell,\mu),
 \nonumber\\
 &
 Q_1(z,\mu):=
 \left( \big(z+\sqrt{z^2-i}\big)^2+ \frac{e^{-2i \mu\sqrt{z^2-i}}}{\big(z+\sqrt{z^2-i}\big)^2} \right)
\left( \big(z+\sqrt{z^2+i}\big)^2+ \frac{e^{-2i \mu\sqrt{z^2+i}}}{\big(z+\sqrt{z^2+i}\big)^2} \right),
\nonumber
 \\
 &
 Q_2(z,\ell,\mu):=e^{-4i\mu\ell z}
 \big(1-e^{-2i\mu\sqrt{z^2-i}}\big)\big(1-e^{-2i\mu\sqrt{z^2+i}}\big).
\nonumber
\end{align}
Consider $z$ obeying the restrictions
$|z|\geqslant r\geqslant \frac{17}{10}$,
$\IM z_2\leqslant 0$ and denote $r_1:=\big(r+\sqrt{r^2-1}\big)^{-1}$.
Then
\begin{equation}
\begin{aligned}
&z+\sqrt{z^2\pm i}=2z\left(1\pm \frac{i}{z^2}\frac{1}{1+\sqrt{1\pm i z^{-2}}}\right),
\qquad \sqrt{z^2\pm i}-z=\pm\frac{i}{z}\frac{1}{1+\sqrt{1\pm i z^{-2}}},
\\
&\big|1+\sqrt{1\pm i z^{-2}}\big|\geqslant 1+\sqrt{1-r^{-2}}\geqslant 1,
\qquad
\big|\sqrt{z^2\pm i}-z\big|\leqslant r_1, 
\\
&\big|z+\sqrt{z^2\pm i}\big|\geqslant 2|z|\left(1-\frac{r_1}{r}
\right)=2|z|\frac{\sqrt{r^2-1}}{r}\geqslant 2\sqrt{r^2-1}.
\end{aligned}\label{3.6a}
\end{equation}
These inequalities allow us to estimate the functions $Q_1$, $Q_2$:
\begin{align}
&
\begin{aligned}
&\big|Q_1(z,\mu)\big|\geqslant \left(
\frac{4(r^2-1)}{r^2}|z|^2
- \frac{e^{2\mu\left(\IM z+ r_1
\right)}}{4(r^2-1)}
\right)^2\geqslant \left(4(r^2-1)- \frac{e^{
2\mu r_1
}}{4(r^2-1)}
\right)^2,
\\
&\big|Q_2(z,\ell,\mu)\big|\leqslant e^{4\mu\ell\IM z}\left(1+e^{2\mu\left(\IM z+ r_1
\right)}\right)^2\leqslant \Big(1+e^{
2\mu r_1
}\Big)^2,
\end{aligned}\label{3.24}
\end{align}
Since $r\geqslant\frac{39}{20}$, $r\geqslant \tfrac{\mu}{2}$, we have
$4(r^2-1)\geqslant \rho_0:=\tfrac{1121}{100}$, $e^{
2\mu r_1
}
\leqslant e^{\frac{156}{39+\sqrt{1121}}}$ and
\begin{align*}
\big|Q(z,\ell,\mu)\big|\geqslant
\big|Q_1(z,\ell,\mu)\big|- \big|Q_1(z,\ell,\mu)\big|\geqslant
\left(\frac{1121}{100}- \frac{100}{1121}e^{\frac{156}{39+\sqrt{1121}}}\right)^2 -  \Big(1+e^{\frac{156}{39+\sqrt{1121}}}\Big)^2>16>0.
\end{align*}
Hence, for the considered values of $z$, we have $|Q(z,\ell,\mu)|>0$ and such $z$ can not be a zero of $Q$.

Let $|z|>r$, $\IM z>0$. Then it follows from (\ref{3.24}) that
\begin{align*}
\big|Q(z,\ell,\mu)\big|\geqslant \left(\frac{4484}{1521}|z|^2-\frac{96}{125}
e^{2\mu\IM z}\right)^2-93e^{4\mu(\ell+1)\IM z}
\end{align*}
and the function $Q$ is non-zero provided $|z|>\frac{47}{25}e^{\mu(\ell+1)\IM z}$.
Assume now that $e^{(\ell+1)\mu \IM z}\geqslant  |z|$, $\IM z>0$,  $|z|>r$. In the same way as above we  estimate
\begin{align*}
\big|Q(z,\ell,\mu)\big| \geqslant&  e^{4\mu\ell\IM z}\left(e^{2\mu\left(\IM z+r_1
\right)}-1\right)^2
 -\left(
4|z|^2\left(1+\frac{r_1}{r}
\right)^2+\frac{e^{2\mu\left(\IM z+r_1 
\right)}}{4(r^2-1)}
\right)^2
\\
\geqslant& 57.827 e^{4\mu(\ell+1)\IM z} -   \left(5.213|z|^2 + 0.768e^{2\mu   \IM z} \right)^2>59 e^{4\mu(\ell+1)\IM z} - 55  |z|^4\geqslant 0
\end{align*}
and the function $Q(z,\ell,\mu)$ is again non-zero.

Consider $z$ obeying $|z|\geqslant r$, $\IM z>0$. Employing the estimate $\sinh t \leqslant t\cosh t$, $t\geqslant 0$, by straightforward calculations for all $w=w_1+i w_2\in \mathds{C}$ we check that
\begin{align*}
|e^w-1|^2=4e^{w_1} \left(\sinh^2 \frac{w_1}{2} + \sin^2 w_2\right)\leqslant \frac{w_1^2}{4}(e^{w_1}+1)^2 + w_2^2\leqslant \frac{7|w|^2 e^{2|w|}}{4}.
\end{align*}
Hence,
\begin{equation*}
\big|e^{-2i \mu\sqrt{z^2\pm i}}-1\big|\leqslant \sqrt{7}\mu \sqrt{|z|^2+1} e^{\mu\sqrt{|z|^2+1}} \leqslant \sqrt{7(1+r^{-2})}\mu |z| e^{\mu\sqrt{|z|^2+1}}.
\end{equation*}
Rewriting   the exponents in the functions $Q_1$, $Q_2$ as $e^{-2 i\mu\sqrt{z^2\pm i}}=1+(e^{-2 i\mu\sqrt{z^2\pm i}}-1)$, by (\ref{3.6a}) we get
\begin{align*}
\big|Q(z,\ell,\mu)\big| \geqslant \left( \frac{73}{25}|z|^2 - \frac{8}{125}\mu |z|e^{\frac{9}{8}\mu |z|}\right)^2 - \frac{177}{20}\mu^2|z|^2 e^{(4\ell+\frac{9}{4})\mu|z|}>0 \quad\text{as}\quad |z| \geqslant \frac{76}{73}\mu e^{(2\ell+\frac{9}{8})\mu |z|}.
\end{align*}

Assume   that $\g$ is small enough. Then according Lemma~\ref{lm3.2}, the function $F$ has only two zeroes in the circle $|k|\leqslant 2\sqrt{\g}$, one of them is $k=0$, while the other is located in the upper half-plane. 
\end{proof}

In what follows, given a set $S$, by $\overline{S}$ we denote the closure of this set, while for a number $z$, the symbol $\overline{z}$ stands for the complex conjugation.

The next  lemmata describe the behavior the zeroes as $\ell$ grows.

\begin{lemma}\label{lm3.4}
Let $\g\ne 0$. For each $n\in\mathds{Z}$ obeying the estimate
\begin{equation}\label{3.17}
\frac{\pi |n|}{2}+\frac{\pi}{4} \leqslant \Tht(\g)\ell,\qquad \Tht(\g):= \frac{2(1-e^{-\frac{\pi}{4}})\g}{17\big(1+\sqrt{1-e^{-\frac{\pi}{4}}}\big) \sqrt{\g}+\frac{10\pi}{\sqrt{3}}},
\end{equation}
the circle $\big\{k:\, |k-\tfrac{\pi n}{2\ell}|<\tfrac{\pi}{4\ell}\}$ contains exactly one simple zero $k_n(\ell,\g)$ of equation (\ref{2.10}). This zero is holomorphic in $\ell^{-1}$ and as $\ell\to+\infty$,  the leading terms of its Taylor series are
\begin{equation}\label{3.4}
\begin{aligned}
k_n(\ell,\g)=& \frac{\pi n}{2} \ell^{-1} -\frac{\sqrt{2} \pi n(\sinh \sqrt{2\g}-\sin\sqrt{2\g})}{4\sqrt{\g}(\cosh\sqrt{2\g}-\cos\sqrt{2\g})}\ell^{-2}
\\
&- \frac{  i\pi^2 n^2 \sinh \sqrt{2\g}\sin\sqrt{2\g} - 2\pi n(\sinh \sqrt{2\g}-\sin\sqrt{2\g})^2}{\g(\cosh \sqrt{2\g}-\cos\sqrt{2\g})^2}\ell^{-3}
+O(\ell^{-4}).
\end{aligned}
\end{equation}
\end{lemma}

\begin{proof}
We divide equation (\ref{2.10}) by $\frac{\sin\sqrt{k^2-i\g}}{\sqrt{k^2-i\g}} \frac{\sin\sqrt{k^2+i\g}}{\sqrt{k^2+i\g}}$ and rewrite it as
\begin{align}\label{3.16}
&1-e^{-4i k\ell}+k\g^{-1} F_1(k,\g) + k^2\g^{-2}F_2(k,\g)=0,
\\
&F_1(k,\g):=2 \left(\cos\sqrt{k^2+i \g}\frac{\sqrt{k^2-i \g}}{\sin\sqrt{k^2-i \g}} - \cos\sqrt{k^2-i \g}\frac{\sqrt{k^2+i \g}}{\sin\sqrt{k^2+i \g}}\right),\nonumber
\\
&F_2(k,\g):=-4 \big( \sqrt{k^2-i \g} \cot\sqrt{k^2-i \g}-i k\big) \big( \sqrt{k^2+i \g} \cot\sqrt{k^2+i \g}-i k\big).\nonumber
\end{align}
Denote $\Pi:=\big\{k:\, |k|<\tfrac{\pi}{4\sqrt{3}}\big\}$. We first estimate the functions $F_1$ and $F_2$ by obtaining appropriate bounds for each term in their definitions. Since $\overline{\sqrt{k^2-i \g}}=\sqrt{\overline{k}^2+i\g}$,
\begin{equation}\label{3.35}
\max\limits_{k\in\overline{\Pi}} |\sqrt{k^2+ i \g} \cot\sqrt{k^2+ i \g}|=\max\limits_{k\in\overline{\Pi}} |\sqrt{k^2- i \g} \cot\sqrt{k^2- i \g}|=\max\limits_{k\in\overline{\Pi}} \left|\sqrt{k^2- i \g}  \frac{\cos\sqrt{k^2\pm i \g}}{\sin\sqrt{k^2\mp i \g}}
\right|
\end{equation}
and it is sufficient to estimate the first maximum only.

Consider first the values $\g\leqslant \tfrac{\pi^2}{24}$. Then
$\big\{\sqrt{k^2+ i\g}:\, k\in\overline{\Pi}\big\}\subset\Pi^{(1)}:=\left\{z:\, |z|\leqslant \frac{\pi}{4}\right\}$.
Hence, the function $k\mapsto \sqrt{k^2+ i \g} \cot\sqrt{k^2+ i \g}$ is holomorphic on $\overline{\Pi^{(1)}}$ and by the maximum modulus principle,
\begin{equation}\label{3.31}
\max\limits_{k\in\overline{\Pi}} |\sqrt{k^2+ i \g} \cot\sqrt{k^2+ i \g}| < \left(\max\limits_{z\in\p\Pi^{(1)}} |z|^2 \frac{|\cosh z|^2}{|\sinh z|^2}
\right)^\frac{1}{2}.
\end{equation}
For $z\in  \p\Pi^{(1)}$ we let $z=\frac{\pi}{4}e^{i t}$, $t\in[0,2\pi)$, and employing the estimates
\begin{equation*}
\sinh y\geqslant y,\quad y\in[0,+\infty),\qquad \sin y\geqslant \frac{2\sqrt{2}}{\pi} y,\quad y\in\left[0,\frac{\pi}{4}\right],\qquad \cos y\leqslant 1-\frac{4}{\pi^2}y^2,\quad y\in\left[0,\frac{\pi}{2}\right],
\end{equation*}
 by straightforward calculations we get
\begin{equation}\label{3.32}
\begin{aligned}
\max\limits_{z\in\p\Pi^{(1)}} |z|^2 \frac{|\cos z|^2}{|\sin z|^2}= &
\frac{\pi^2}{16}\max\limits_{t\in[0,\tfrac{\pi}{2})} \left(1+
\frac{\cos\left(\frac{\pi}{2}\cos t\right)}{\sinh^2\left(\frac{\pi}{4}\sin t\right)+\sin^2\left(\frac{\pi}{4}\cos t\right)}\right)
\\
\leqslant & \frac{\pi^2}{16}\max\limits_{t\in[0,\tfrac{\pi}{2})} \left(1+
\frac{\sin^2 t}{\sinh^2\left(\frac{\pi}{4}\sin t\right)+\frac{1}{2}\cos^2 t}\right)=\frac{\pi^2}{16\sinh^2\frac{\pi}{4}}.
\end{aligned}
\end{equation}

We proceed to the case $\g>\tfrac{\pi^2}{24}$. Here the set
$\big\{k^2+i\g:\, k\in\Pi\big\}$ is the circle of radius
$\frac{\pi^2}{48}$ centered at $i\g$ and
hence,
\begin{align*}
&\big\{\sqrt{k^2+i\g}:\, k\in \Pi\big\} \subset  \left\{z:\, \sqrt{\g-\tfrac{\pi^2}{48}}<|z|<\sqrt{\g+\tfrac{\pi^2}{48}},\,\right.
\\
&\left.\hphantom{\big\{\sqrt{k^2+i\g}:\, k\in \Pi\big\} \subset }
\tfrac{\pi}{24}<\tfrac{\pi}{8}-\tfrac{1}{2}\arcsin\tfrac{\pi^2}{48\g}<\mathrm{arg}\, z<\tfrac{\pi}{8}+\tfrac{1}{2}\arcsin\tfrac{\pi^2}{48\g}<\tfrac{5\pi}{24}
\right\}\subset \Pi^{(2)}
\\
&\Pi^{(2)}:= \left\{z:\, \sqrt{\g-\tfrac{\pi^2}{48}}\cos\tfrac{5\pi}{24}<\RE z
<\sqrt{\g+\tfrac{\pi^2}{48}}\cos\tfrac{\pi}{24},\,
\sqrt{\g-\tfrac{\pi^2}{48}}\sin\tfrac{\pi}{24}<\IM z
<\sqrt{\g+\tfrac{\pi^2}{48}}\sin\tfrac{5\pi}{24}
\right\}.
\end{align*}
Again by the maximal modulus principle we get:
\begin{align*}
&\max\limits_{k\in\overline{\Pi}} |\sqrt{k^2+ i \g} \cot\sqrt{k^2+ i \g}|
=\sqrt{\g+\frac{\pi^2}{48}} \left(\max\limits_{z_1+iz_2\in\Pi^{(2)}}
\frac{\cosh 2z_2+\cos 2z_1}{
\cosh 2z_2-\cos 2z_1}
\right)^\frac{1}{2}
\\
&\leqslant \sqrt{\g+\frac{\pi^2}{48}} \left(\max\limits_{
z_1+iz_2\in\Pi^{(2)}
%
}  \frac{\cosh 2z_2 + 1}{\cosh 2z_2 - 1}
\right)^\frac{1}{2}
=
\sqrt{\g+\frac{\pi^2}{48}}\left(
  \frac{\cosh 2  \sqrt{\g-\frac{\pi^2}{48}}\sin\frac{\pi}{24} - 1}{\cosh 2 \sqrt{\g-\frac{\pi^2}{48}}\sin\frac{\pi}{24} - 1}
\right)^\frac{1}{2}. 
\end{align*}
This estimate and (\ref{3.35}), (\ref{3.31}), (\ref{3.32}) yield:
\begin{equation*}
\max\limits_{k\in\overline{\Pi}} |F_1(k,\g)|\leqslant 17\sqrt{\g+\frac{\pi^2}{48}},
\qquad
\max\limits_{k\in\overline{\Pi}} |F_2(k,\g)|\leqslant \frac{1}{4}
\left(17\sqrt{\g+\frac{\pi^2}{48}}+\frac{\pi}{\sqrt{3}}\right)^2.
\end{equation*}
Employing  these inequalities, by straightforward calculations we check that
\begin{equation}\label{3.36}
|k\g^{-1}F_1(k,\g)+k^2\g^{-2} F_2(k,\g)|<
 1-e^{-\frac{\pi}{4}}\quad \text{as}\quad |k|
\leqslant \Tht(\g).
\end{equation}

We choose $n\in\mathds{Z}$ obeying (\ref{3.17}) and denote
$\Pi_n:=\left\{k:\, \left| k-\frac{\pi n}{2\ell}\right|<\frac{ \pi}{4\ell}\right\}$.
In each $\Pi_n$, the function $k\mapsto 1-e^{-4i k\ell}$ has the only zero $k=\frac{\pi n}{2\ell}$ and this zero is simple.
It is straightforward to check that for $k\in \Pi_n$ the identities hold:
$\min\limits_{k\in\p\Pi_n} |1-e^{-4i k\ell}|=\min\limits_{|z|=1} |1-e^{\frac{\pi}{4}z}|=1-e^{-\frac{\pi}{4}}$.
As $k\in\p\Pi_n$, estimate (\ref{3.36}) is also satisfied,  and by
 Rouch\'e theorem equation (\ref{3.16}) possesses  the only zero in $\Pi_n$. We denote this root by $k_n(\ell,\g)$.

Let us find out the asymptotic behavior of $k_n$ as $\ell\to+\infty$. We denote   $\xi:=\ell k$, $\e:=\ell^{-1}$, $k=\e\xi$ and rewrite equation (\ref{3.16}) as
\begin{equation}\label{3.22}
1-e^{-4i\xi} + 2\e \g^{-1} F_1(\e\xi,\g) +\e^2\xi^2 \g^{-2} F_2(\e\xi,\g)=0.
\end{equation}
This equation has the root $\xi_n(\e,\g):=\e^{-1} k_n(\ell,\g)$ for each $n$ obeying (\ref{3.17}). The function in its left hand side is jointly holomorphic in $\e$, $\xi$ and $\g$ and
$(1-e^{-4i\xi})'\big|_{\xi=\frac{\pi n}{2}}=4i\ne0$.
Hence, by the inverse function theorem, the root $\xi_n(\e)$ is jointly holomorphic in $\e$ and $\g$. We write the leading terms of the Taylor expansion of $\xi_n(\e,\g)$.
\begin{equation*}
\xi_n(\e,\g)=\frac{\pi n}{2}+\xi_n^{(1)}(\g)\e+\xi_n^{(2)}(\g)\e^2+O(\e^3).
\end{equation*}
We substitute this expansion into  equation (\ref{3.22}), expand the result in $\e$ and equate the coefficients at the like power of $\e$. This determines 
the coefficients  $\xi_n^{(1)}$, $\xi_n^{(2)}$ and proves (\ref{3.4}).
\end{proof}

It is straightforward to confirm that the zeroes of the functions $F_\pm$ correspond to the eigenvalues and resonances of equation (\ref{2.1}) as $V=V_\pm$ is a single step function:
\begin{equation}\label{3.23}
V_\pm(x):=\left\{
\begin{aligned}
\pm i&\g\quad\text{as}\quad \pm x\in[0,1],
\\
&0\quad\text{otherwise}.
\end{aligned}\right.\nonumber
\end{equation}

\begin{lemma}\label{lm3.6}
The functions $F_\pm$ have countably many zeroes accumulating at infinity only. There are finitely many zeroes of these functions in the lower half-plane and all of them are located in the ball $\{k:\, |k|<r_0\sqrt{\g}\}$, $r_0:=\frac{6}{5}\max\big\{1,\sqrt{\g}\big\}$. All zeroes of $F_\pm$ are simple except a possible zero $k=2 i$.
If $k=2i$ is a zero of $F_+$ or $F_-$, it is of the second order. For all $\g>0$, the function $F_\pm$ can have at most one real zero,
which  is negative for $F_+$ and is positive for $F_-$.
\end{lemma}

\begin{proof}
The functions $F_\pm$ are entire and this is why they have countably many zeroes accumulating at infinity only.  In the same way how estimate (\ref{3.24}) was proved, we immediately get
\begin{equation*}
|F_\pm(k,\g)|\geqslant
4(r_0^2-1)- \frac{e^{
2\sqrt{\g} r_2
}}{4(r_0^2-1)}>0\qquad\text{as}\qquad |k|\geqslant \sqrt{\g}r_0,\quad \IM k\leqslant0,\quad r_2:=\Big(r_0 +\sqrt{r_0^2+1}\Big)^{-1}.
\end{equation*}
Hence, all zeroes of $F_\pm$ with negative imaginary part are located in the circle $|k|<\sqrt{\g}r_0$. Since they   accumulate at infinity only, there are only finitely many of them in the aforementioned circle.

Let $k$ be a zero of $F_+$. It is clear that $k\ne0$, $\sqrt{k^2+i\g}\ne 0$. We calculate the first derivative and we see immediately that if $\frac{\p F_+}{\p k}(k,\g)=0$, by the equation $F_+(k,\g)=0$ this implies   $k=2i$. Then we calculate the second derivative of $F_+$ at $k=2i$ and see that it is nonzero. Hence, all zeroes of $F_+$ are simple except a possible second order zero at
$k=2i$. The case of $F_-$ is studied in the same way.

To find   real zeroes of $F_+$, similar to the proof of Lemma~\ref{lm3.5}, for real positive $k$  we make the change
\begin{equation}\label{4.1}
k=\frac{\b}{2} \sqrt{u^{-2}-u^2},\quad 0<u<1,
\qquad\sqrt{k^2\pm i \g}=\frac{\sqrt{\g}}{\sqrt{2}} (u^{-1}\pm i u),\qquad i\g\pm 2k^2=\g (i\pm u^{-2}\mp u^2).
\end{equation}
Taking then the real and imaginary parts of $F_+$ and factorizing them, we obtain:
\begin{equation}\label{3.7}
\begin{aligned}
&
\left(u^2 \cos\frac{\b}{2u}-\sqrt{1-u^4}\sin\frac{\b}{2u}\right) \left(\sinh\frac{\b u}{2}+\sqrt{1-u^2}\cosh\frac{\b u}{2}\right)=0,
\\
&
\left(u^2 \sin\frac{\b}{2u}+\sqrt{1-u^4}\cos\frac{\b}{2u}\right) \left(\cosh\frac{\b u}{2}+\sqrt{1-u^2}\sinh\frac{\b u}{2}\right)=0.
\end{aligned}
\end{equation}
The second factors in the above equations are obviously non-zero and the first ones can not vanish simultaneously since otherwise we would have got $\cos\frac{\b}{2u}=\sin\frac{\b}{2u}=0$. To find negative real zeroes, we make the same change as in (\ref{4.1}) but with $\b$ replaced by $-\b$. Then we arrive immediately to equations (\ref{3.7}), where the same replacement is made. Since $\sinh\frac{\b u}{2}<\cosh\frac{\b u}{2}$, only the first factor in the second equation can vanish and then only the second factor is the first equation can be zero:
\begin{equation}\label{3.8}
\sqrt{1-u^4}\cos\frac{\b}{2u}-u^2 \sin\frac{\b}{2u}=0,\qquad \sqrt{1-u^4}=\tanh \frac{\b u}{2}.
\end{equation}
The function $\sqrt{1-u^4}$, $u\in(0,1)$, decreases monotonically from $1$ to $0$, while $\tanh\frac{\b u}{2}$ increases monotonically from $0$ to $\tanh\frac{\b}{2}$. Hence, the second equation in (\ref{3.8}) has the unique root. This root solves the first equation only for certain values of $\b$. The zeroes of $F_-$ are studied in the same way.
\end{proof}

Let $k=K^{(n)}_\pm(\g)$, $n=1,\ldots,N_\pm$, be the zeroes of the functions $F_\pm$ in the lower half-plane. According Lemma~\ref{lm3.6}, all of them are simple. Consider the points $K_\pm^{(n)}$ on the complex plane identifying coinciding ones. Let  $d$ be the minimal among the mutual distances between these points  and the distances from them to the real axis.
We introduce the constants:
\begin{align*}
&C_1:=\g^2+(1+r^2)^{-1}\sinh^2 \sqrt{(1+r^2)\g},
\quad && C_2:=\inf\limits_{\substack{
|\RE k|\leqslant\sqrt{\g} r,
\\
-\frac{d}{2}<\IM k<0
}} |F_-(k,\g)F_+(k,\g)| |\IM k|^{-2},
\\
&  C_3:=\min\limits_{
\substack{
|k|=\sqrt{\g}r,
\\
\IM k\leqslant 0
}}
|F_-(k,\g)F_+(k,\g)|,
&&
C_4:=\min\limits_{n}\min\limits_{|k-K^{(n)}_\pm|\leqslant \frac{d}{2}}  |F_-(k,\g)F_+(k,\g)| |k-K^{(n)}_\pm|^{-2}.
\end{align*}

\begin{lemma}\label{lm3.7}
As
\begin{equation}\label{3.25}
\ell>\frac{\sqrt{C_2}}{\sqrt{C_1}}\max\left\{1,
\left(\frac{C_1}{C_3}\right)^\frac{1-\ln 2}{2}
\right\},\qquad
\frac{\ln \frac{\ell\sqrt{C_1}}{\sqrt{C_2}}}{2(1-\ln 2)\ell}\leqslant \frac{d}{2},
\end{equation}
the function $F$ has exactly $N_-+N_+$ zeroes $k^{(n)}(\ell,\g)$, $n=1,\ldots,N_\pm$, counting their orders in the half-plane
\begin{equation*}
\IM k\leqslant -\frac{\ln \frac{\ell\sqrt{C_1}}{\sqrt{C_2}}}{2(1-\ln 2)\ell}.
\end{equation*}
The constant $C_4$ is positive and for $\ell$ obeying
 obeying (\ref{3.25}) and
 \begin{equation}\label{3.26}
\ell>\frac{1}{2\IM K_\pm^{(n)}+ d} \ln\frac{2\sqrt{C_1}}{d\sqrt{C_4}},
\end{equation}
the estimates hold:
\begin{equation}\label{3.27}
\big|k^{(n)}_\pm(\ell,\g)-K^{(n)}_\pm(\g)\big|<\frac{\sqrt{C_1}}{\sqrt{C_4}} e^{(2\IM K_\pm^{(n)}+d)\ell}.
\end{equation}
\end{lemma}

\begin{proof}
We begin with the estimate
\begin{equation}\label{3.28}
\g^2\left|\frac{\sin\sqrt{k^2-i\g}}{\sqrt{k^2-i\g}} \frac{\sin\sqrt{k^2+i\g}}{\sqrt{k^2+i\g}}
\right|\leqslant C_1\quad\text{as}\quad |k|\leqslant \sqrt{\g}r,\quad \IM k\leqslant 0,
\end{equation}
which can be checked by straightforward calculations employing the maximum modulus principle. According Lemma~\ref{lm3.3}, all zeroes of the function $F$ below the line $\IM k=-\frac{\ln \frac{\ell\sqrt{C_1}}{\sqrt{C_2}}}{2(1-\ln 2)\ell}$ also obey $|k|<\sqrt{\g}r$ and therefore, they are contained in
\begin{equation*}
\Pi:=\left\{k:\, |k|<\sqrt{\g}r,\; \IM k <-\frac{\ln \frac{\ell\sqrt{C_1}}{\sqrt{C_2}}}{2(1-\ln 2)\ell}\right\}.
\end{equation*}
As $k\in\p\Pi$, by the definition of constants $C_i$, $i=1,\ldots,4$, and (\ref{3.25}), (\ref{3.28}), the inequalities
\begin{equation*}
|F_-(k,\g) F_+(k,\g)|\geqslant \min \left\{
\frac{C_2\ln^2 \frac{\ell\sqrt{C_1}}{\sqrt{C_2}}}{4(1-\ln 2)^2\ell^2}, C_3\right\} > \max\limits_{\overline{\Pi}} \g^2\left|e^{-4ik\ell} \frac{\sin\sqrt{k^2-i\g}}{\sqrt{k^2-i\g}} \frac{\sin\sqrt{k^2+i\g}}{\sqrt{k^2+i\g}}\right|
\end{equation*}
hold.
Hence, by the Rouch\'e theorem we immediately conclude that the function $F$ has the same amount of zeroes inside $\Pi$ as the function $F_-F_+$ does and this proves the first part of the lemma.

The constant  $C_4$ is non-zero since the functions $F_\pm$ has a simple zero  at $K^{(n)}_\pm$ and the function $F_- F_+$ has also simple zeroes at $K^{(n)}_\pm$ if $K^{(n)}_-\ne K^{(m)}_+$ and a second order zero if $K^{(n)}_- = K^{(m)}_+$.
Then by inequalities  (\ref{3.28}) and the definition of the constant  $C_4$
we infer that, as
\begin{equation*}
|k-K^{(n)}_\pm|=\rho:=\frac{\sqrt{C_1}}{\sqrt{C_4}} e^{-(2\IM K_\pm^{(n)}-d)\ell}<\frac{d}{2}
\end{equation*}
and $\ell$ obeys (\ref{3.25}), (\ref{3.26}),
 the estimates
\begin{equation*}
|F_-(k,\g) F_+(k,\g)|\geqslant C_4 \rho^2, \qquad
\g^2\left|e^{-4ik\ell} \frac{\sin\sqrt{k^2-i\g}}{\sqrt{k^2-i\g}} \frac{\sin\sqrt{k^2+i\g}}{\sqrt{k^2+i\g}}\right|\leqslant C_1 e^{(4\IM K^{(n)}_\pm+2d)\ell}<C_4\rho^2
\end{equation*}
are valid.  Hence, we can apply the Rouch\'e theorem to the function $F$ and we conclude that this function has the same amount of zeroes inside the circle $\{k:\, |k-K^{(n)}_\pm|<\rho\}$
as the functions $F_- F_+$ does. This implies estimates (\ref{3.27}).
\end{proof}

Finally, we prove the absence of  pure imaginary zeroes in the lower complex half-plane, which means that the waveguide does not have   real discrete eigenvalues $\lambda<0$.

\begin{lemma}\label{lm3.5}
Equation (\ref{2.10}) has no pure imaginary zeroes in the lower complex half-plane.
\end{lemma}

\begin{proof}
For negative pure imaginary $k$ we make the change
\begin{equation*}
k=-\frac{i\b}{2}\sqrt{u^{-2}-u^2},\quad 0<u<1,\quad \sqrt{k^2\pm i \g}=\frac{\b}{2} (i u^{-1}\pm u),
\quad
i\g\pm 2k^2=\frac{\b^2}{2} (i\mp u^{-2}\pm u^2),
\end{equation*}
where $\b:=\sqrt{2\g}$. We multiply equation (\ref{2.10}) by $\sqrt{k^2-i\g}\sqrt{k^2+i\g}$,
substitute the above formulae  and divide the result by $\sqrt{1-u^4}$. This leads us to  the equation:
\begin{align*}
u^{-4}\sqrt{1-u^4}\cosh \b u^{-1}&+\frac{1-e^{-\frac{2\b\ell}{u}\sqrt{1-u^4}}}{2\sqrt{1-u^4}} (\cosh \b u^{-1}-\cos \b u)
\\
&+
u^{-4}\sinh \b u^{-1} +\sqrt{1-u^4}\cos\b u - u^2\sin \b u=0.
\end{align*}
The   first two  terms in this equation are non-zero for all $u$, $\b$, $\ell$. Therefore, to prove the lemma, it is sufficient to check that the sum of remaining three terms is non-negative.
For $\b>\arcsinh 1$
we have
\begin{align*}
&\sqrt{1-u^4}\cos\b u - u^2\sin \b u\leqslant \Big(\sqrt{1-u^4}^2+u^4\Big)^\frac{1}{2}=1,
\\
&u^{-4}\sinh \b u^{-1} +\sqrt{1-u^4}\cos\b u - u^2\sin \b u > u^{-4}\sinh \b u^{-1}-1>\sinh \b -1> 0.
\end{align*}
As
$0\leqslant \b\leqslant \arcsinh 1<\frac{\pi}{2}$, by
the inequalities
$\cos \b u\geqslant \cos \b>0$, $u^2\sin\b u\leqslant \sin \b$,
$\sinh\b-\sin\b\geqslant 0$,
we get
\begin{equation*}
u^{-4}\sinh \b u^{-1} +\sqrt{1-u^4}\cos\b u - u^2\sin \b u \geqslant u^{-4}\sinh \b - \sin \b\geqslant  \sinh\b-\sin\b\geqslant 0.
\end{equation*}
The proof is complete.
\end{proof}

\section{
Behavior of zeroes and ladders of resonances and eigenvalues}
\label{sec:zeros}

In this section we employ the results of the previous section to describe how the zeroes of the function $F$ depend  on $\g$ and $\ell$.

\subsection{Location of zeroes and dependence of gain-and-loss amplitude}

According Remark~\ref{rm1}, as $\g=0$, there is a single zero $k=0$. It corresponds to a resonance and the associated solution to equation (\ref{2.1}) is the constant function. As soon as $\g$ is positive, no matter how small it is, there arise {\sl infinitely many} zeroes of the function $F$, see Lemma~\ref{lm3.1}, and we still have a resonance at $k=0$ but now the associated solution of equation (\ref{2.1}) is non-constant. Despite   \textcolor{black}{the potential $V$ is a {\sl regular} perturbation as $\g$ is small}, it \textcolor{black}{influences}   the general spectral picture quite essentially producing at once infinitely many zeroes. All these zeroes are located symmetrically with respect to the imaginary axis and are continuous in $\g$ and $\ell$. For the corresponding eigenvalues and resonances given by $\l=k^2$, this property   means that we deal with complex-conjugate pairs of resonances and eigenvalues, as it is usual for $\mathcal{PT}$-symmetric systems.

By Lemma~\ref{lm3.2}, 
for small $\g$, we have one more pure imaginary zero close to $k=0$. This zero is holomorphic in $\g$ and the leading terms of its Taylor series are given by (\ref{3.1}).
Since $(\ell+\tfrac{1}{3})\g>0$, this zero corresponds to a resonance. The corresponding value $\l$
is negative and   \textcolor{black}{lies close or next to the bottom} of the essential spectrum.

Lemmata~\ref{lm3.1} and~\ref{lm3.3} state the following important facts. The zeroes can accumulate at infinity only, that is, each circle of a fixed radius in the complex plane contains only finitely many zeroes. Except for zeroes in the circle $|k|<\sqrt{\g} r$, where $r$ is introduced in (\ref{3.5a}), all other zeroes are located in the upper half-plane in the exponential sector
\begin{equation*}
\sqrt{\g}e^{(\ell+1)\IM k}<|k|<\frac{47}{25}
 \sqrt{\g}e^{(\ell+1)\IM k},\qquad |k|<\frac{76}{73}\g e^{(\ell+\frac{9}{8})|k|}.
\end{equation*}
The second inequality deserves a special consideration. Let $\g$ be small enough and consider the zeroes outside the circumference $|k|=\sqrt{\g} r$. Then the inequality implies
\begin{equation*}
\frac{39}{20}<\frac{76}{73}\g^\frac{1}{2} e^{(\ell+\frac{9}{8})|k|}\qquad \Rightarrow\qquad |k|>\frac{1}{\ell+\frac{9}{8}}\ln \frac{2847\g^{-\frac{1}{2}}}{1520}
\end{equation*}
and hence,  except the zeroes discussed in Lemma~\ref{lm3.2}, all other zeroes are located far from $k=0$ at distances of order $O(|\ln\g|)$. This means that for small $\g$, infinitely many zeroes emerge from infinity in the upper half-plane and only two of them come from zero and   are  located in the upper half-plane, too.

Since all zeroes with $\IM k \leq 0$   are located in the aforementioned circle $|k|<\sqrt{\g} r$,   for all   $\g>0$ and $\ell\geqslant 0$, there are only {\sl finitely many} such zeroes.  They 
correspond to  the eigenvalues (bound states), and for them and inequality (\ref{3.0b}) should be satisfied as well.  All eigenvalues   are complex-valued since according Lemma~\ref{lm3.5}, there are no pure imaginary zeroes in lower complex half-plane.  Real zeroes 
correspond to spectral singularities. Thus,  for all $\ell\geqslant 0$ and $\g>0$ the system has infinitely many resonances  and   finitely many bound states and spectral singularities.    The size of the aforementioned circle, in which the zeroes with  $\IM k \leq 0$  are located, increases proportionally to $\g$ (see the definition of $r$ in (\ref{3.5a})). Hence, by increasing $\g$ we have more chances to get the zeroes in the lower half-plane and as we shall discuss below, this is indeed the case.

\subsection{Large distance regime and ladders of eigenvalues and resonances}\label{sec:ladders}

Here we discuss   the behavior of the zeroes  as $\ell$ grows. Lemma~\ref{lm3.7} states that  below the line $\IM k=-\frac{\ln \frac{\ell\sqrt{C_1}}{\sqrt{C_2}}}{2(1-\ln 2)\ell}$, the function $F$ possesses the same number zeroes as the functions $F_-$ and $F_+$ do and these  zeroes of $F$ converge to the zeroes of $F_\pm$ as $\ell\to+\infty$, see (\ref{3.27}). As $\ell\to+\infty$, the distance from the aforementioned line to the real axis is of order $O(\ell^{-1}\ln\ell)$ and is small. A natural question whether there can be  some other zeroes in the lower half-plane above this line is answered in Lemma~\ref{lm3.4} and the answer is positive. According to this lemma,  making $\ell$ large enough, we \textcolor{black}{certainly} have $2N+1$ zeroes in the vicinity of the real axis, where $N:=\lfloor \tfrac{1}{2}+\tfrac{2\Tht(\g)}{\pi}\ell\rfloor$, $\lfloor\cdot\rfloor$ is the integer part of a number. All these zeroes are located in the circles
$\big\{k:\, |k-\tfrac{\pi n}{2\ell}|<\tfrac{\pi}{4\ell}\big\}$, exactly one zero in each circle, and all these circles are inside a fixed circle $\{k:\,|k|<\Tht(\g)\}$. As $\ell$ grows, the number  of such zeroes increases as well but all of them \textcolor{black}{remain} in the circle $\{k:\,|k|<\Tht(\g)\}$. As $\g$ grows, the size of the latter circle increases as $O(\sqrt{\g})$.   For large $\ell$, these zeroes are approximately given by asymptotic formulae (\ref{3.4}). As these formulae show, if $\sin\sqrt{2\g}<0$, the zeroes are located in the upper half-plane and correspond to  resonances.  The described behavior is similar to that of Wannier-Stark resonances. Namely, for gain-and-loss amplitude we again have a ladder of resonances accumulating in the vicinity certain zone in the essential spectrum.  As $\sin\sqrt{2\g}>0$, these zeroes are in the lower half-plane and correspond to eigenvalues. The   number of these zeroes  guaranteed by  Lemma~\ref{lm3.4}  grows as $O(\ell)$ as $\ell\to+\infty$, while the  distances between the neighbouring zeroes are of order $O(\ell^{-1})$. This means that we have either resonances or eigenvalues accumulating to the segment $[-\Tht(\g),\Tht(\g)]$ on the real axis, see Figure~\ref{fig:accum}(a,b). It is especially interesting to notice that by means of a continuous  change  of the gain-and-loss amplitude one can transform a ladder of resonances  to a ladder of eigenvalues or vice versa as shown in figure~\ref{fig:accum}(b). Such a transformation is obviously not possible for  Wannier-Stark ladders   in  self-adjoint operators where complex eigenvalues cannt exist.  The peculiar regime  $\sin\sqrt{2\g}=0$ corresponds to the situation when the leading order of imaginary part of expansion (\ref{3.4}) vanishes [curve~2 in figure~\ref{fig:accum}(b)]. A   delicate analysis shows that in this case the imaginary parts of zeroes described by  Lemma~\ref{lm3.4} are  nonzero and  amount to  $O(\ell^{-5})$. Thus, for $\sin\sqrt{2\g}=0$ we have a ladder of coexisting  ``nearly spectral singularities'' with extremely small by nevertheless nonzero imaginary parts (genuine spectral singularities with exactly zero imaginary parts  will be discussed in section~\ref{sec:ss}).

\begin{figure}
	\centering
\includegraphics[width=0.99\columnwidth]{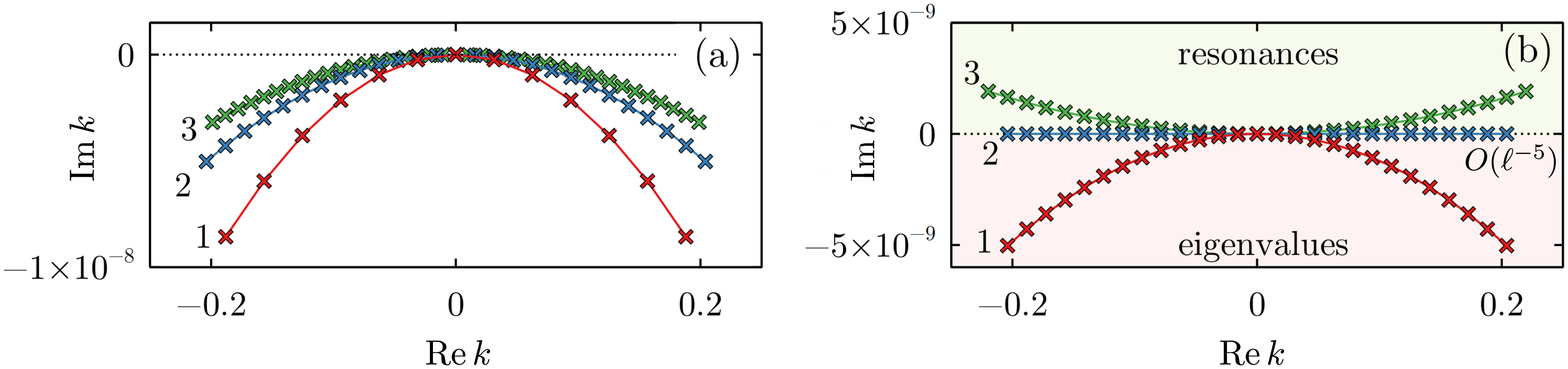}
\caption{Ladders of zeroes  plotted according to the leading terms of  expansion (\ref{3.4}). 	(a) Accumulation of zeroes (corresponding to eigenvalues) to the segment $[-\Tht(\g),\Tht(\g)]$ as the gain-to-loss distance $\ell$ grows. Here the gain-and-loss amplitude $\gamma=40$ and $\ell=50$ (curve 1), 100 (curve 2) and  150 (curve 3). The zeroes are shown with crosses; thin solid lines are to guide the eye. (b) Transformation of the ladder of eigenvalues (curve 1) to a ladder of resonances (curve 3) under the change of the gain-and-loss amplitude. Here $\ell=100$ and $\gamma=40$ (curve 1) and $\gamma = 50$ (curve 3). Curve 2 corresponds to $\gamma=9\pi^2/2\approx  44.413$. In this case the imaginary part of term $\propto \ell^{-3}$ in  expansion (\ref{3.4}) vanishes, and we obtain a ladder of coexisting ``nearly spectral singularities'' with Im$\,k=O(\ell^{-5})$.}
\label{fig:accum}
\end{figure}

It is interesting to compare this spectral picture with that in the limiting situation $\ell=+\infty$ and the first issue is which operator are to be regarded as limiting ones. Such limiting operator are to be treated in the norm resolvent sense, that is, we should find an approximation for the resolvent of our operator as $\ell\to+\infty$.
According the results of \cite{UMJ},  as $\ell\to+\infty$, the solution of the equation
\begin{equation*}
-u''+Vu-\l u=f\quad\text{in}\quad\mathds{R},\quad u\in H^2(\mathds{R}),\quad f(x):=f_-(x+\ell)+f_0(x)+f_+(x-\ell),\quad f_0,\,f_\pm\in L_2(\mathds{R}),
\end{equation*}
is approximated in the norm resolvent sense up to a small error by $u_-(x+\ell)+u_0(x)+u_+(x-\ell)$. Here $u_0, u_\pm\in H^2(\mathds{R})$ solve the equations
\begin{equation*}
-u_0''-\l u_0=f_0,\quad -u_\pm''+V_\pm u_\pm -\l u_\pm=f_\pm\quad\text{in}\quad \mathds{R},
\end{equation*}
where the potentials $V_\pm$ were introduced in (\ref{3.23}).
Hence, the limiting operator is in fact a direct sum of three operators
$\Op_-\oplus\Op_0\oplus\Op_+$, $ \Op_0:=-\frac{d^2\ }{dx^2}$,
$\Op_\pm:=-\frac{d^2\ }{dx^2}+V_-$.
The essential spectra of all three operators are $[0,+\infty)$. The operator $\Op_0$ has the only resonance at $k=0$, while the resonances and the eigenvalues of the operators $\Op_\pm$ are expressed via the zeroes of the functions $F_\pm$ by the formula $\l=k^2$, see Lemma~\ref{lm3.6}.  In view of such approximation, it would be natural to expect that as $\ell\to+\infty$, the above discussed zeroes $k_n(\ell,\g)$   converge to similar real zeroes  of the functions $F_\pm$ or to zero.
{\sl This is not true!} Indeed, by Lemma~\ref{lm3.6}, each of the functions $F_\pm$ has at most {\sl one}  real zero, while the number of the zeroes $k_n(\ell,\g)$ in the vicinity of this segment {\sl increases} as $\ell\to+\infty$. Thus, regarding the case $\ell=+\infty$ and eigenvalues equations for the operators $\Op_\pm$ as limiting, for sufficiently large finite $\ell$ we get a {\sl large number} of {\sl eigenvalues or resonances emerging from internal points} of the zone $[0,\Tht^2(g)]$ in essential spectrum. And at most three points (including zero)  are spectral singularities for the limiting equations. In other words, here we have eigenvalues emerging from ordinary points in the essential spectrum not being spectral singularities for the unperturbed operators. And while in the case of Wannier-Stark resonances their existence is due to the eigenvalues of the operator $-\frac{d^2\ }{dx^2}+\e x$ accumulating on the real line as $\e\to+0$, this surely not the case in our model.

Here we just make one more important observation. The presence of the mentioned zero-width resonances, or real zeroes of $F$, could be regarded as an explanation or the reason for the existence of the above described ladder of resonances and eigenvalues. However, {\sl this is not the case!} Indeed, for $\frac{\pi^2}{2} < \gamma< \gamma_*\approx  11.561$, the function $F$ {\sl does have} a ladder of resonances, but as we shall show in the next section, it {\sl has no} real zeroes. So, our ladder of resonances or eigenvalues exists due to some other reasons and  the existence of the zero-width resonances is mostly the implication of the presence of such ladder and its behavior as $\g$ varies.

\section{Spectral singularities}
\label{sec:ss}
\subsection{General analytical expressions}

In this section we study real zeroes of the function $F$ corresponding to spectral singularities, i.e. to  zero-width resonances.
For such zeroes, equation (\ref{2.10}) is a pair of two real equations for one real variable $k$ and two parameters $\ell$ and $\g$.
Thanks to the symmetry of the zeroes with respect to the imaginary axis, it is sufficient to find only positive real resonances since the negative ones are located symmetrically with respect to the origin. Similar to the proof of Lemma~\ref{lm3.6}, for real positive $k$  we make  change (\ref{4.1})
and rewrite equation (\ref{2.10}) in the following form:
\begin{align*}
(1-2u^{-4})\cos\b u^{-1}&+(2u^4-1)\cosh\b u + 2i \sqrt{1-u^4}\left(u^2\sinh\b u - u^{-4}\sin\b u^{-1}\right)
\\
&-e^{-2i\b\ell\sqrt{u^{-2}-u^2}} \left(\cosh(\b u)-\cos\b u^{-1}\right)=0,\qquad \b=\sqrt{2\g}.
\end{align*}
Taking the real and imaginary part of this equation and multiplying the equation by $u^4$, we obtain:
\begin{equation}\label{4.3}
\begin{aligned}
&(u^4-2)\cos\b u^{-1}+u^4(2u^4-1)\cosh\b u
=u^4\cos\left(2\b\ell\sqrt{u^{-2}-u^2}\right)\left(\cosh \b u-\cos\b u^{-1}\right),
\\
&2\sqrt{1-u^4}\Big(u^6\sinh\b u-\sin\b u^{-1}\big)
=-u^4\sin\left(2\b\ell\sqrt{u^{-2}-u^2}\right)\left(\cosh \b u-\cos\b u^{-1}\right).
\end{aligned}
\end{equation}
This is a system of two real equations with three real variables. If we are given $(\b,\ell)$ and we try to find $u$, the system is overdetermined and does not necessary have a root. In other words, it is solvable with respect to $u$ only if $(\b,\ell)$ are located on some (solvability) curves. In order to avoid working with an overdetermined system, in what follows we regard (\ref{4.3}) as a system for two unknown variable with one parameter.

To find the curves in $(\b,\ell)$ plane, on which equations (\ref{4.3}) are solvable with respect to $u$, we shall regard $u$ as a parameter and $(\b,\ell)$ as unknown variables. We take the sum of squares of equations (\ref{4.3}) and divide the result by $(1-u^4)$. We also divide equations (\ref{4.3}). This leads us to a   pair  of equations:
\begin{align}\label{4.4}
&u^4(1-u^4) \cosh\b u\cos\b u^{-1}-2u^6 \sinh\b u\sin\b u^{-1} +1-u^{12}=0,
\\
&\frac{2\sqrt{1-u^4}(u^6\sinh\b u-\sin\b u^{-1})}{(2-u^4)\cos\b u^{-1}+u^4(1-2u^4)\cosh\b u}=\tan \Big( 2\b\ell\sqrt{u^{-2}-u^2}\Big).\nonumber
\end{align}
The second equation can be solved explicitly with respect to $\ell$:
\begin{equation}\label{4.10}
\ell=\frac{1}{2\b\sqrt{u^{-2}-u^2}} \left(\arctan \frac{2\sqrt{1-u^4}(u^6\sinh\b u-\sin\b u^{-1})}{(2-u^4)\cos\b u^{-1}+u^4(1-2u^4)\cosh\b u}+\pi n\right),
\end{equation}
where $n\in\mathds{N}$ is an arbitrary natural number. As the next lemma states, to make equations (\ref{4.4}), (\ref{4.10}) equivalent to (\ref{4.3}), we should also assume that
\begin{equation}\label{4.5}
(-1)^n=\sign \big((u^4-2)\cos\b u^{-1}+u^4(2u^4-1)\cosh\b u\big).
\end{equation}

\begin{lemma}\label{lm4.1}
Equations (\ref{4.3}) are equivalent to (\ref{4.4}), (\ref{4.10}), (\ref{4.5}).
\end{lemma}

\begin{proof}
We rewrite shortly equations (\ref{4.3}) as
$A_1=B\cos\a$, $ A_2=-B\sin\a$,
where $A_1$, $A_2$ are the left hand sides in (\ref{4.3}), $\a=2\b\ell\sqrt{u^{-2}-u^2}$ and $B=u^4(\cosh\b u-\cos\b u^{-1})$. Then equations (\ref{4.4}), (\ref{4.10}) become
$A_1^2+A_2^2=B^2$, $\a=-\arctan\frac{A_2}{A_1}+\pi n$.
We have:
\begin{equation*}
B\cos \a=(-1)^n\cos\arctan\frac{A_2}{A_1}=\frac{(-1)^n}{\sqrt{1+\frac{A_2^2}{A_1^2}}} = \frac{(-1)^n|A_1|}{\sqrt{A_1^2+A_2^2}}
\end{equation*}
and we get the first equation $A_1=B\cos\a$  provided condition (\ref{4.5}) is satisfied. In the same way we check that the latter condition also ensures the second equation $A_2=-B\sin\a$.
\end{proof}

Equation (\ref{4.4}) is transcendental, and we can not solve it analytically.
Nevertheless, for each $u\in(0,1)$, this is an equation only for a single variable $\b$, not for two as equations (\ref{4.3}). So, we propose the following algorithm of recovering the aforementioned solvability curves: choose $u\in(0,1)$, then solve equation (\ref{4.4}) and recover the sequence of  distances  $\ell$ by formula (\ref{4.10}) with different integer $n$. Then the gain-and-loss amplitude $\gamma$ and the corresponding wavenumber $k$ can be readily recovered from $\beta$ and  $u$.  In a similar way, one can first fix some value of $\beta$ (i.e., fix the gain-and-loss strength) and then   solve equation (\ref{4.4}) with respect to $u$ and   recover   $\ell$ by (\ref{4.10}). Equation (\ref{4.4})  is well-behaved, and for each $\beta$  all its zeros  $u$ can be easily found numerically.

Alternatively, as explained below  in Section~\ref{sec:general}, the values  corresponding to  spectral singularities can be found systematically by means of the numerical continuation from the limit $\ell=0$. However, in this case   equation (\ref{4.4}) is still useful because it allows one to check that all spectral singularities have been found for the given value of the gain-and-loss $\gamma$.

\subsection{
Absence of   spectral singularities}
\label{sec:gap}

For $u=0$ and $u=1$, the left-hand-side of equation (\ref{4.4}) is equal respectively to $1$ and $-2\sinh\beta\sin\beta$. Then a sufficient condition for the existence of a spectral singularity at the given gain-and-loss amplitude $\gamma$ is $\sin\beta=\sin\sqrt{2\gamma}>0$. At the same time, it is also possible to establish sufficient conditions that forbid the existence of   spectral singularities in a certain interval of parameters. In this subsection we prove the existence of two ``forbidden gaps'' for the roots of equation (\ref{4.4}). The first one exists for all $\b\geqslant 0$ and it states that there is no roots in certain interval. The second gap is a certain interval of values of $\b$, for which equation (\ref{4.4}) has no zeroes at all.

For the convenience, by $g(u,\b)$ we denote the left hand side of equation (\ref{4.4}). The first ``forbidden gap'' is described in the following lemma.

\begin{lemma}\label{lm5.1}
For all $\b\geqslant 0$, equation (\ref{4.4}) has no roots in the interval $\big[0,(1+\tfrac{\b}{4})^{-1}\big)$.
\end{lemma}
\begin{proof}
Employing a standard inequality $a\cos\a+b\sin\a\leqslant \sqrt{a^2+b^2}$, we estimate the first two term in equation (\ref{4.4}) as
\begin{equation*}
u^4(1-u^4) \cosh\b u\cos\b u^{-1}-2u^6 \sinh\b u\sin\b u^{-1} \leqslant \sqrt{u^8(1-u^4)^2\cosh^2 \b u+4u^{12}\sinh^2\b u}.
\end{equation*}
Hence, equation (\ref{4.4}) surely has no roots for values of $u$ satisfying
\begin{equation*}
 \sqrt{u^8(1-u^4)^2\cosh^2 \b u+4u^{12}\sinh^2\b u} <1-u^{12}.
\end{equation*}
Expressing $\cosh^2 \b u$ via $\sinh^2 \b u$ and simplifying this inequality, we obtain $ u^4(1+u^4)\cosh \b u <1+u^{12}$ and hence,
\begin{equation}\label{3.37}
\cosh \b u-1 <\frac{1-2u^4+u^8}{u^4},\qquad \sqrt{2}\sinh \b u
<u^{-2}-u^2.
\end{equation}
Thanks to the estimate $\sqrt{2}\sinh 2\ln u^{-1}\leqslant u^{-2}-u^2$, $u\in(0,1]$, the last inequality in (\ref{3.37}) holds once $\frac{\b}{4}<u^{-1}\ln u^{-1}$. It is easy to confirm that this inequality is true once $u^{-1}>1+\tfrac{\b}{4}$.
The proof is complete.
\end{proof}

The next lemma is   auxiliary  and will be employed in studying the second forbidden zone.

\begin{lemma}\label{lm5.2}
The function $g(u,\pi)$ is positive on $[0,1)$.
\end{lemma}

\begin{proof}
We have $g(0,\pi)=1$ and by Lemma~\ref{lm5.1}, it is positive for $u<(1+\tfrac{\pi}{4})^{-1}$. This is why in what follows we consider only the values $u\geqslant (1+\tfrac{\pi}{4})^{-1}$. For such values of $u$ we have
$\pi\leqslant \pi u^{-1}\leqslant \pi (1+\tfrac{\pi}{4})<1.79\pi$.
As
$\tfrac{3\pi}{2} \leqslant \pi u^{-1}\leqslant \pi (1+\tfrac{\pi}{4})$,
the function $\sin\b u^{-1}$ is negative, while $\cos \b u^{-1}$ is positive. Hence, for such values of $u$, the function $g(u,\pi)$ is positive. It remains to consider the values $\frac{2}{3}<u\leqslant 1$.

For such values of $u$ we first observe the following simple estimates:
\begin{equation*}
\frac{1-u^4}{1-u}\cos\pi u^{-1}\geqslant -4,\qquad -\frac{\sin \pi u^{-1}}{1-u}=\frac{\sin \pi (u^{-1}-1)}{1-u}\geqslant \pi u.
\end{equation*}
Then the function $g$ satisfies the estimate:
\begin{equation}\label{3.12}
\frac{g(u,\pi)}{u^4(1-u)}\geqslant g_1(u)+g_2(u),\qquad g_1(u):=-4\cosh\pi u +2\pi u^3 \sinh \pi u,\qquad g_2(u):=\frac{1-u^{12}}{u^4(1-u)}.
\end{equation}
The function $g_1(u)$ is monotonically increasing in $u\in[\tfrac{2}{3},1]$ since
\begin{equation*}
g_1'(u)=2\pi^2 u^3\cosh\pi u+(6\pi u^2-4\pi)\sinh\pi u>2\pi(\pi u^3+6u^2-2)\sinh\pi u>2\pi\sinh\pi u>0.
\end{equation*}
Hence,
$g_1(u)\geqslant g_1\left(\frac{2}{3}\right)>-9.05$.
For the function $g_2$ we have the following representation and estimate:
\begin{equation*}
g_2(u)=1+\sum\limits_{j=1}^{4}(u^j+u^{-j})+\sum\limits_{j=5}^{8} u^j\geqslant 9 +\sum\limits_{j=5}^{8}\left(\frac{2}{3}\right)^j>9.31.
\end{equation*}
Two last estimates and (\ref{3.12}) imply the positivity of the function $g$ for $u\in[\tfrac{2}{3},1)$.
\end{proof}

Denote
\begin{align*}
g_*(u,\b):=&\b u^3(1-3u^4)\cosh\b u \sin \b u^{-1}+\b u^5 (3-u^4)\sinh\b u \cos \b u^{-1}
\\
&-2 u^4(1+u^4) \cosh\b u \cos \b u^{-1} -6(1+u^{12}).
\end{align*}

The next lemma states the existence of the second forbidden zone.

\begin{lemma}\label{lm5.3}
Equation (\ref{4.4}) has no roots as $\pi<\b<\b_*<5$, where $(u_*, \b_*)$ is  the root of the system of the equations
\begin{equation}\label{3.13}
g(u,\b)=0,\qquad g_*(u,\b)=0,\qquad u\in[0,1],\qquad \pi<\b<5,
\end{equation}
with minimal possible $\b$.
Their approximate values are
\begin{equation}\label{3.14}
\b_*=4.808438,\qquad u_*=0.611772.
\end{equation}
\end{lemma}

\begin{proof}
The function $g(u,\pi)$ is positive on $[0,1)$ and $g(1,\pi)=0$, see Figure~\ref{fig:forbidden}a. As $\pi<\b<2\pi$, we have $g(0,\b)=1>0$ and $g(1,\b)=-2\sin\b\sinh\b>0$. Hence, for $\b$ close enough to $\pi$, the function $g(u,\b)$ is positive for all $u\in[0,1]$. At the same time, we have $g(0.65,5)<-0.617<0$ and therefore, for $\b=5$, equation (\ref{4.4}) possesses at least two roots, one in $(0,0.65)$ and another in $(0.65,1)$. cf. Figure~\ref{fig:forbidden}c. We also observe that the function $g$ is jointly continuous in $(u,\b)$. The above facts means that as $\b$ grows from $\pi$ to $5$, at some value $\b=\b_*$, the graph of the function $g$ is still located in the upper half-plane but touches the $u$-axis at some point $u=u_*$,
see Figure~\ref{fig:forbidden}b. The function $g(u,\b)$ is positive as $\pi<\b<\b_*$ and $u\in[0,1]$. Then the point $u=u_*$ is obviously the global minimum of $g$ and hence, $(u_*,\b_*)$ is a solution to the system of equations $g(u,\b)=0$, $\frac{\p g}{\p u}(u,\b)=0$. It is easy to check that $g_*=\frac{\p g}{\p u}-6 g$ and hence, $(u_*,\b_*)$ solves system (\ref{3.13}). These roots can be found numerically and this gives (\ref{3.14}). The proof is complete.
\end{proof}

\begin{figure}
	\centering
\includegraphics[width=0.99\columnwidth]{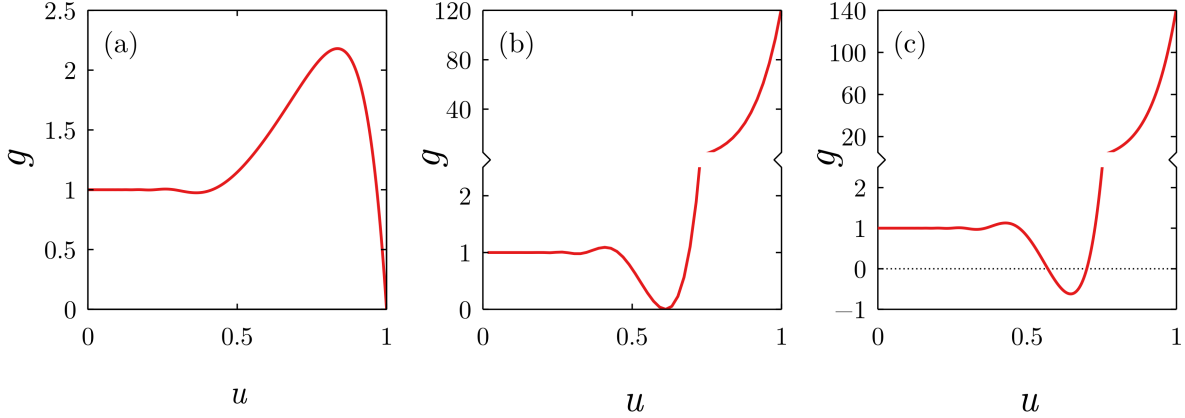}
\caption{Illustration for proof of Lemma~\ref{lm5.3}. Graphs of the function $g(u, \beta)$ for $\b=\pi$ (a), $\b=\b_*\approx  4.808438$ (b), and $\b=5$ (c). Notice broken vertical axes in (b) and (c). }
\label{fig:forbidden}
\end{figure}

Returning from the auxiliary variable $\beta$ to the gain-and-loss amplitude $\gamma=\beta^2/2$, from   Lemma~\ref{lm5.3}  we deduce  the following important result:
\begin{equation}
\label{eq:gap}
\textrm{there is no spectral singularities for\quad }  \frac{\pi^2}{2} < \gamma< \gamma_*\approx  11.561.
\end{equation}

\subsection{Creating a  spectral singularity  at a given wavenumber}

\begin{table}
	\centering
	\begin{tabular}{c|cccc}
		$n$ & 0 & 1& 2&\\\hline
		$\gamma_\star^{(n)}$ & 2.071&13.307&27.783 &\\[2mm]
		$k_\star^{(n)}$ & 1.065&4.318 &7.529 &
	\end{tabular}
	\caption{Approximate values of gain-and-loss amplitudes $\gamma=\gamma_\star^{(n)}$  and wavenumbers   $k=k_\star^{(n)}$, $n=0, 1, \ldots$,  corresponding to spectral singularities with lowest $\gamma$ in the limit $\ell=0$, see Table~I in  \cite{Mostafazadeh2009}. \label{tbl:1}}
\end{table}

For $\ell=0$ spectral singularity can only be obtained for some isolated values of the wavenumber $k$ and the gain-and-loss amplitude $\gamma$ \cite{Mostafazadeh2009}. Several lowest values of $\gamma$ corresponding to the spectral singularities  and the associated wavenumbers  $k$ are listed in Table~\ref{tbl:1}.  An  important advantage of the more general system with nonzero gain-to-loss disctance    $\ell>0$ consists  in the possibility to create a spectral   singularity   at any wavenumber $k$ given beforehand. Indeed, let us  return back to equations (\ref{4.4}), (\ref{4.10}) and    discuss the following   issue: given a point $k$ on the real axis, how to choose $\b$ and $\ell$ to have a  resonance at this point? Equations (\ref{4.4})--(\ref{4.10}) allow us to answer  easily this question.

We fix $k>0$ and we find the associated value of $u$ by resolving (\ref{4.1}):
\begin{equation}\label{4.12}
u^{-2}-u^2=4k^2\b^{-2},\qquad
u=\b R^{-1}, \quad \mbox{where\ } R = \sqrt{2k^2+\sqrt{4k^4+\b^4}}.
\end{equation}
We divide equation (\ref{4.4}) by $u^6$ and substitute then the above formulae and
\begin{equation*}
\frac{1-u^{12}}{u^6}=\frac{1-u^4}{u^2}\frac{1+u^4+u^8}{u^4}=(u^{-2}-u^2)\big((u^{-2}-u^2)^2+3
\big).
\end{equation*}
This gives the equation:
\begin{equation}\label{4.11}
\begin{aligned}
2\b^4 k^2&\cosh(\b^2R^{-1})  \cos R
 - \b^6 \sinh(\b^2R^{-1})  \sin R
  +2k^2(16k^4+3\b^4)=0.
\end{aligned}
\end{equation}
An algorithm for creating a resonance at a prescribed point  $k$ is as follows. Given $k>0$, we first solve equation (\ref{4.11}) with respect to $\b$ and we also find $u$ by (\ref{4.12}). Then needed values of  $\ell$ are determined by (\ref{4.10}), (\ref{4.5}).

In order to illustrate this  algorithm,  we us consider a finite interval of wavenumbers $k\in (0, k_1]$, where  we set $k_1 = 10$ for the numerics reported on in what follows. We scan   the chosen interval  with a sufficiently small step ($\Delta k=0.01$) and for each value of $k$ solve equation (\ref{4.11}) numerically  using the simple dichotomy method. While for each $k$  equation (\ref{4.11}) might have several roots $\beta$, in our numerical procedure we always choose the  minimal positive root, i.e., the one which allows to achieve the spectral singularity with given $k$ at the smallest possible value  of the gain-and-loss amplitude $\gamma=\gamma_\textrm{min}$. Next, we choose the minimal positive distance $\ell_\textrm{min}$ which satisfies the conditions   (\ref{4.10}), (\ref{4.5}) and then we use the periodicity in $\ell$ to generate the sequence of larger gain-to-loss distances $\ell_n = \ell_\textrm{min} + n\pi/(2k)$, $n=1,2\ldots$ [see   (\ref{eq:periodic})]. The resulting dependencies $\gamma_\textrm{min}(k)$ and $\ell_\textrm{min}(k)$, $\ell_n(k)$   are shown in figure~\ref{fig:min}. The minimal gain-and-loss amplitude $\g_{min}(k)$  and the minimal gain-to-loss distance $\ell_{min}$ are   discontinuous,  which means that the small variation  in  the wavenumuber $k$  might require  a significant  change either in   $\gamma$ or in $\ell$.  It is especially important that the values of the  distance  $\ell$ are generically different from zero, which points out explicitly    that the new degree of freedom offered by the nonzero gain-to-loss distance    is   important for the achieving a spectral singularity  at   the given wavenumber $k$.

\begin{figure}
	\centering
	\includegraphics[width=0.8\columnwidth]{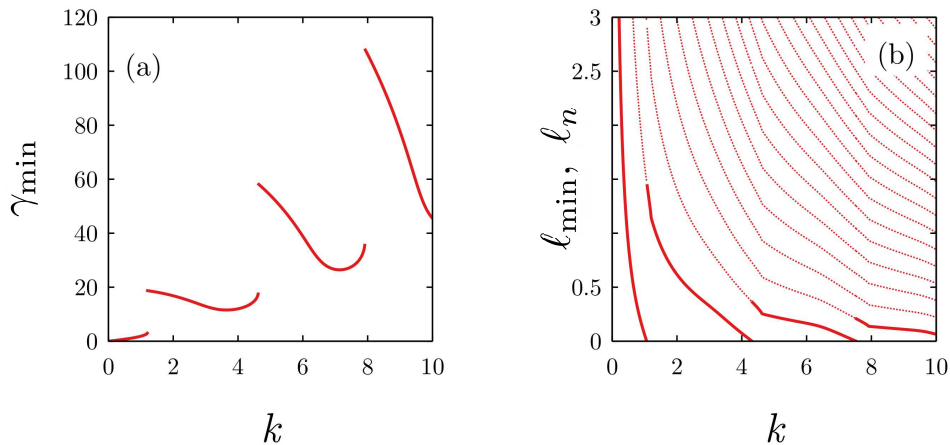}%
	\caption{(a) Minimal value of the gain-and-loss amplitude $\gamma_\textrm{min}$  which corresponds to   a spectral singularity with the given value of the wavenumber $k$. (b)  Minimal gain-to-loss distance  $\ell_\textrm{min}$  which corresponds to   a spectral singularity with the  $k$ and $\gamma$ from the  left panel (bold curves) and larger distances $\ell_n$ obtained using the periodicity in $\ell$ (thin \textcolor{black}{dotted} curves).}
	\label{fig:min}
\end{figure}

\subsection{$\PT$-symmetry breaking laser-antilaser threshold}

A particularly important characteristics of any $\PT$-symmetric structure is the $\PT$ symmetry breaking threshold, i.e., the amplitude of the gain-and-loss
corresponding to the  ``phase transition'' from the purely real to complex spectrum.  The best studied scenario of the phase transition  
is the collision of two real discrete  eigenvalues at an exceptional point with the subsequent splitting in a complex-conjugate pair. However, in systems with nonempty continuous spectrum,  the phase transition can also occur through   the splitting of a self-dual spectral singularity, which results in a bifurcation of a complex-conjugate pair from an interior point of the continuum \cite{Yang17,KZ17,Konotop2019}. At the moment corresponding to the formation of the spectral singularity, the system operates in the CPA-laser regime \cite{Longhi10}. Thus, in such a system, the $\PT$-symmetry breaking threshold at the same time corresponds to the CPA-laser threshold.

Lemma~\ref{lm3.3} guarantees that the spectrum of our system is real for sufficiently small gain-and-loss amplitudes $\gamma$. Additionally, according to Lemma~\ref{lm3.5}, the spectrum  does not have any real discrete eigenvalue. Hence, the $\PT$-symmetry breaking is expected to occur through the emergence of a self-dual spectral singularity. In order to identify the $\PT$-symmetry breaking threshold in our system, we start from the limit $\ell=0$, there the phase transition takes place at  $\gamma_\star^{(0)} \approx 2.072$, see \cite{Mostafazadeh2009,KZ17} and Table~\ref{tbl:1}. Thus, the spectrum with $\ell=0$ is purely real and continuous for $\gamma \in [0, \gamma_\star]$, while the  increase of the gain-and-loss just above $\gamma_\star^{(0)}$ leads to the bifurcation of a complex-conjugate pair from an interior point  of the continuum. The spectral singularity forming at  $\gamma_\star^{(0)}$ takes place at wavenumber $k=k_\star^{(0)} \approx  1.065$. Respectively, the complex-conjugate pair of eigenvalues bifurcates from $\lambda_0 =  [k_\star^{(0)}]^2$. (Notice that the further increase of $\gamma$ above the next threshold values listed in Table~\ref{tbl:1} leads to the formation of new spectral singularities and, respectively, to bifurcations of new complex-conjugate pairs in the spectrum.)

Next, we use the numerical continuation in $\ell$ in order to continue the known solution at $\ell=0$ to the domain $\ell>0$.  The obtained dependence of the threshold value  of the gain-and-loss amplitude on distance $\ell$ is shown in figure~\ref{fig:threshold}(a), where we observe that the phase transition threshold decreases monotonically with the growth of $\ell$. This means that introducing an additional space   between the gain and loss, one can decrease the $\PT$-symmetry breaking  threshold, i.e. achieve the laser-antilaser operation at \textcolor{black}{lower} gain-and-loss amplitudes than in a waveguide with the adjacent gain and loss.

\begin{figure}
	\centering
	\includegraphics[width=0.85\columnwidth]{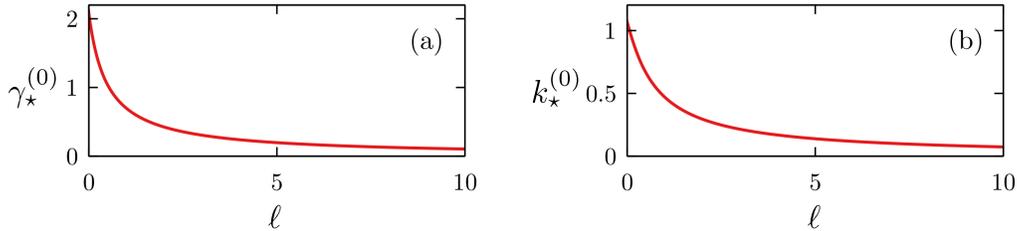}
	\caption{(a) $\PT$-symmetry breaking threshold $\gamma_\star^{(0)}$ \textit{vs} the distance between the gain and loss $\ell$. The spectrum is purely real and continuous for $\gamma\leq \gamma_\star$, but acquires a pair of complex conjugate eigenvalues  as the gain-and-loss amplitude exceeds the threshold $\gamma_\star^{(0)}$. (b) Values of the wavevector $k_\star^{(0)}$ corresponding to the dependence in (a).}
	\label{fig:threshold}
\end{figure}

\subsection{General picture of spectral singularities}
\label{sec:general}

\begin{figure}
	\centering
	\includegraphics[width=0.99\columnwidth]{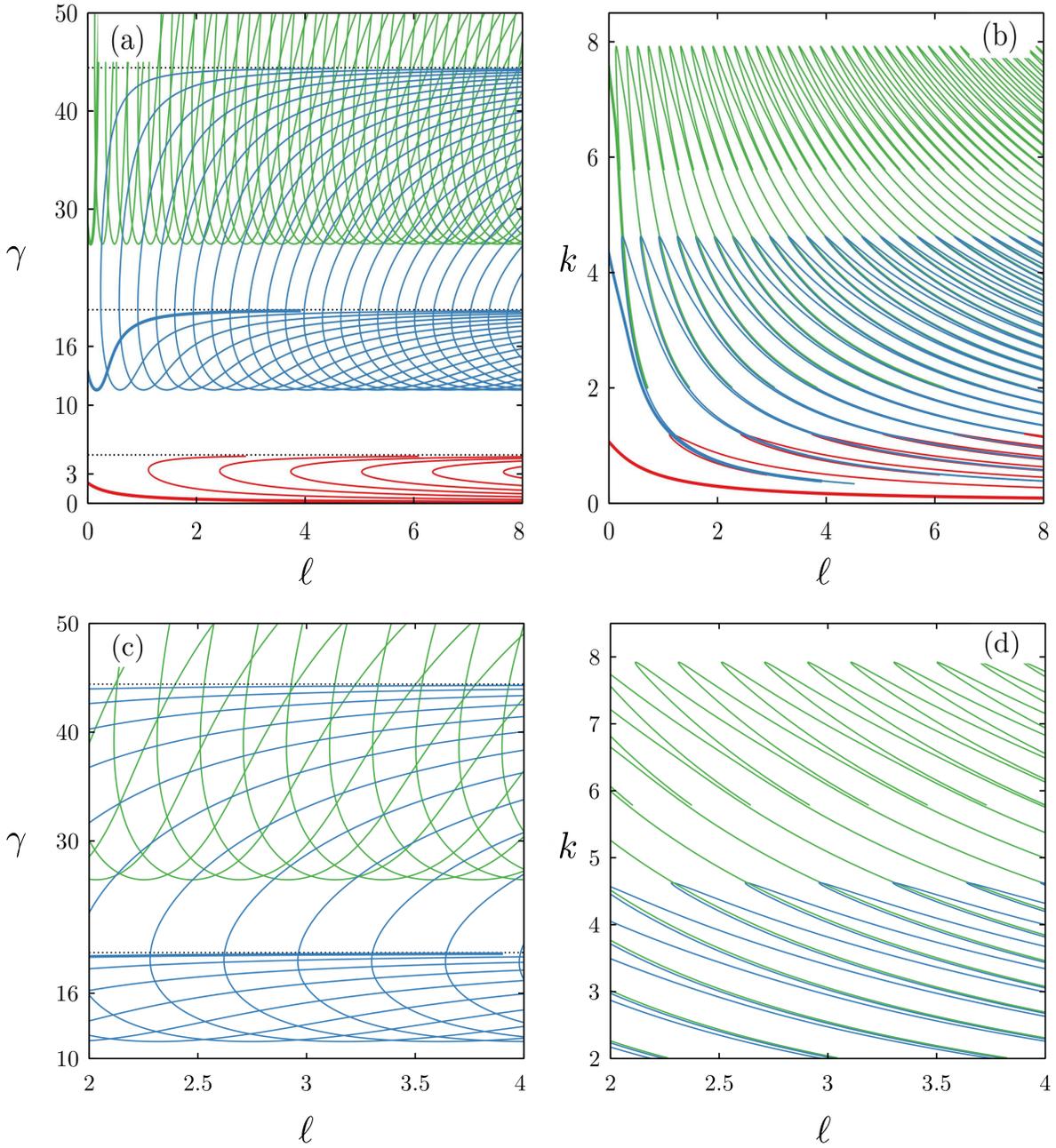}%
	\caption{(a) Values of the gain-and-loss amplitude $\gamma$ and distance  $\ell$, for which   spectral singularities occur. (b) Corresponding  wavenumbers   $k$. \textcolor{black}{Panel (c) magnifies   the region  $(\ell, \gamma)\in [2,4]\times[10,50]$ from   (a), and panel  (d) is the  magnification of corresponding curves from panel (b).}}
	\label{fig:general}
\end{figure}

Let us now turn to description of the general picture of spectral singularities. In order to construct systematically different solutions, we again start from the limit $\ell=0$, where the values of $\gamma$ and $k$ corresponding to spectral singularities are known, see Table~\ref{tbl:1}. Then we use the periodicity of function $e^{-4ik\ell}$ in $\ell$ in order to construct new branches of solutions having no counterparts in the limit $\ell=0$, cf. equation (\ref{eq:periodic}). This procedure results in a fairly complicated picture containing a multitude of spectral singularities, some part of which (corresponding to  relatively  small   values of the gain-and-loss) is shown in figure~\ref{fig:general}(a,b) as the curves on the plane $\gamma$ {\it vs} $\ell$ and $k$ {\it vs} $\ell$.

Let us describe the structure of the found solutions using the  diagram $\gamma$ {\it vs} $\ell$ in figure~\ref{fig:general}(a). The multitude of curves shown in this plot can be divided into three groups (plotted with red, blue and green curves) which have been obtained by means of the continuation from three different solutions in the limit $\ell=0$. There is a well-visible vertical gap between red and blue curves, which corresponds to the ``forbidden'' values of the gain-and-loss amplitudes $\gamma$ found above in (\ref{eq:gap}). At the same time,  there is no gap between blue and green curves, which  results in a multitude of intersections between these curves.
The first group of spectral singularities,  corresponding to red curves in Figure~\ref{fig:general}(a),   is obtained through  the continuation from the  spectral singularity  in the limit $\ell=0$ with the smallest gain-and-loss amplitude, i.e.  from values $\gamma_\star^{(0)}$ and $k_\star^{(0)}$ in Table~\ref{tbl:1}.  The leftmost (bold)  curve in this group which originates in the limit $\ell=0$  is the $\PT$-symmetry breaking threshold which was already shown in figure~\ref{fig:threshold}(a). For values of $\gamma$ above this curve, there is always one or more (but finitely many, see Section~\ref{sec:zeros})   complex-conjugate pairs of eigenvalues in the spectrum.  Several curves situated to the right from the bold red curve are obtained using the fact that if $\gamma$ and $k$ are solutions in the limit $\ell=0$, then the same values of $\gamma$ and $k$ also correspond to a spectral singularity with $\ell_n = (n\pi)/(2k)$, $n=1,2,\ldots$. Once a single solution with a new distance   $\ell_n$ is obtained, a new branch of solutions can be constructed using numerical continuation in $\gamma$ or in $\ell$. In the limit $\ell \to \infty$ all the red curves in figure~\ref{fig:general}(a) approach  the asymptotic value $\gamma=0$ or at $\gamma_*= \pi^2/2\approx 4.935$. Notice that, except for the bold line demarcating the $\PT$-symmetry breaking threshold, none of the red curves can be continued to the limit $\ell\to 0$.

The multitude of red curves  in figure~\ref{fig:general}(a) demonstrates explicitly how new spectral singularities emerge with the increase of $\ell$. Indeed, drawing an imaginary horizontal line, say, at $\gamma=3$ [see the vertical axis  tick  in figure~\ref{fig:general}(a)], we observe that with the increase of $\ell$ this line   intersects more and more red curves. Each intersection corresponds to the values of $\gamma$ and $\ell$ at which   some  root $k$   crosses the real axis and goes down from the upper to lower complex half-plane. Thus,  a finite-width resonance transforms  to a complex eigenvalue through the spectral singularity (we recall again that  in view of   $\PT$ symmetry  any root $k\ne 0$ crosses the real line simultaneously with its counterpart $-\bar{k}$, i.e.  the corresponding   spectral singularity is self-dual).
In order to illustrate this process, in Figure~\ref{fig:g=3}  we show the evolution of three numerically found complex zeros of function $F(k, \gamma, \ell)$  under   the increase of $\ell$.  In this figure the imaginary part of each complex zero changes sign from positive to negative and then asymptotically approaches zero (remaining negative). In the limit of large $\ell$,  this behavior agrees  with expansion (\ref{3.4}) of Lemma~\ref{lm3.4} , where, for the chosen value of $\gamma$,   we have $\sin\sqrt{2\gamma}>0$.   Thus the growing distance   $\ell$ results in a sequence of self-dual spectral singularities and in the increasing number of complex-conjugate eigenvalues in the spectrum.

 \begin{figure}
 	\centering
 	\includegraphics[width=0.8\columnwidth]{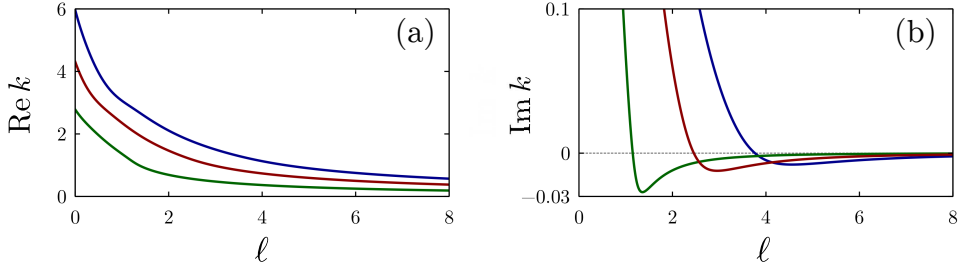}%
 	\caption{(a,b) Real and imaginary parts of three complex zeros of function $F$ for fixed $\gamma=3$ and increasing $\ell$. For each curve, the imaginary part is positive for small $\ell$ and becomes negative for all sufficiently large $\ell$. \textcolor{black}{For each three shown eigenvalues, its real and 	imaginary parts are of the same colour in both panels.}}
 	\label{fig:g=3}
 \end{figure}

The second group of spectral singularities [blue curves in figure~\ref{fig:general}(a)] was obtained by means of the continuation from the  next solution in the limit $\ell=0$, i.e.  from $\gamma_\star^{(1)}$ and $k_\star^{(1)}$ in Table~\ref{tbl:1}. Again, one of the curves [the leftmost bold curve in figure~\ref{fig:general}(a)] was obtained through the direct continuation from the limit $\ell=0$, while other blue curves were generated  using the periodicity of function $F(k, \gamma, \ell)$  in $\ell$ and cannot be continued to the limit $\ell\to0$. In the limit $\ell\to\infty$ the blue curves approach the horizontal asymptotes  $\gamma=2\pi^2\approx19.739$ and $\gamma=9\pi^2/2\approx 44.413$. In  the $(\gamma, \ell)$-plane  the  gain-and-loss amplitudes corresponding to the   group of blue curves are well separated from  those for the red curves: indeed, all red curves are bounded from above by the asymptote $\gamma= \pi^2/2\approx 4.935$, while all blue curves are bounded from below by $\gamma_*\approx 11.561$, see (\ref{eq:gap}). Thus, the  emergence of new spectral singularities with the increase of $\ell$ is sensitive to the value of the gain-and-loss amplitude and does not occur for     the gain-and-loss amplitudes lying in the  gap  between the red and blue curves.

In comparison with the red curves discussed above, the curves from the blue group in figure~\ref{fig:general}(a) feature more complicated behavior and, in particular, can intersect each other (and also intersect the curves from the next, third group of green curves discussed below). The intersections between the blue curves occur for the gain-and-loss amplitudes in the interval $11.561 \lessapprox   \gamma <  9\pi^2/2\approx 19.739$, where $\sin\sqrt{2\gamma}$ is negative.  \textcolor{black}{At first glance,} this might seem to contradict to  the  expansion (\ref{3.4}) of Lemma~\ref{lm3.4},  which  suggests that in this case the multitude of complex zeros   accumulate in the upper complex half-plane with the growth of $\ell$. However, this apparent contradiction is resolved if we trace the behavior of the complex roots more closely.   Indeed, choosing for an example $\gamma =16$ [see the vertical axis tick in  figure~\ref{fig:general}(a)] and computing several  complex roots under the increase of $\ell$, we observe that each considered root first goes down from the upper half-plane to the lower one but then again returns to the upper half-plane, see Figure~\ref{fig:g=16}(b) and (b$_1$). Thus,  in  this  interval of the gain-and-loss amplitudes the increase of $\ell$  results  either in the transformation from the resonance to the eigenvalue or to the opposite process,  i.e. to the disappearance of the complex-conjugate pair.  Respectively, the intersection between the two blue curves corresponds to  two coexisting spectral singularities, i.e. to the moment  when one complex-conjugate pair of eigenvalues disappears and another pair (with different $k$) emerges. For sufficiently large $\ell$ the imaginary part  of each   considered root  remains positive and approaches zero, in accordance with expansion (\ref{3.4}). Thus, in this interval of the gain-and-loss strengths, the limit $\ell\to\infty$ the spectrum contains only a finite number of complex-conjugate eigenvalues, which correspond to complex zeros $k$ whose behavior is not covered   by expansion  (\ref{3.4}).

 \begin{figure}
	\centering
	\includegraphics[width=0.99\columnwidth]{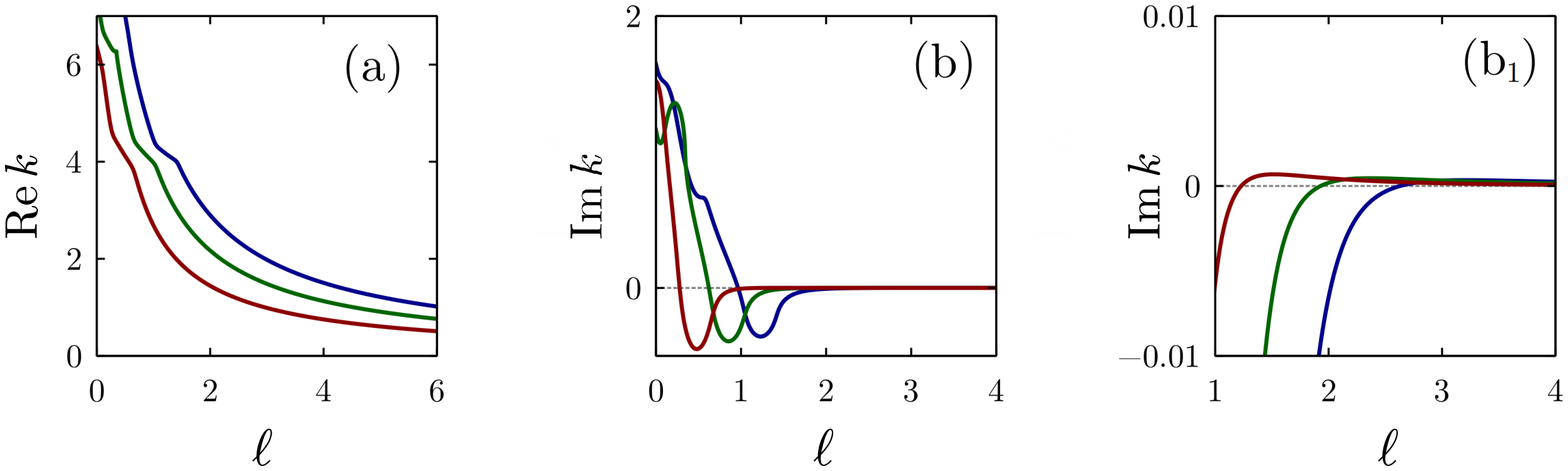}%
	\caption{(a,b) Real and imaginary parts for three   complex zeros of function $F$ for fixed $\gamma=16$ and increasing $\ell$. Panel b$_1$ is the magnification of some region of (b) and  shows more clearly that imaginary parts of all   three shown eigenvalues are positive for all sufficiently large  $\ell$. \textcolor{black}{For each three shown eigenvalues, its real and 	imaginary parts are of the same colour in both panels.}}
	\label{fig:g=16}
\end{figure}

The third  group of spectral singularities [green curves in figure~\ref{fig:general}(a)] was obtained by means of the continuation from  the solution $\gamma_\star^{(2)}$ and $k_\star^{(2)}$ in Table~\ref{tbl:1}. In view of the very complicated structure of the overall resulting picture, we only show a section of these curves corresponding to relatively small   values of the gain-and-loss amplitudes $\gamma$. Quite interestingly, there is no gap between the blue and green groups of curves, which results in the multitude of intersections between blue and green curves, \textcolor{black}{see Fig.~\ref{fig:general}(a) and the magnified view in Fig.~\ref{fig:general}(c)}. These   intersections suggest a possibility of   simultaneous emergence of two complex-conjugate pairs of eigenvalues from two different interior points of the continuous spectra.

Considering further solutions $\gamma_\star^{(n)}$, $k_\star^{(n)}$, $n=3,4,\ldots$, in the limit $\ell=0$ one can construct new groups of spectral singularities  with larger values of $\gamma$, which are not shown in Figure~\ref{fig:general}.

\section{Conclusion}
\label{sec:concl}

A pair of geometrically identical absorptive and active elements is  the simplest   realization of a $\PT$-symmetric system. Properties of $\PT$-symmetric systems are usually considered subject to the change of the gain-and-loss amplitude. In the present work, we have approached the problem from a different perspective and considered how the spectral properties of the $\PT$-symmetric waveguide  depends on the distance between gain and loss elements.   We have derived a compact equation 
for resonances, spectral singularities and bound states and analyzed  it in detail for different gain-to-loss distances.

The first main result of our study shows that in the limit of large gain-to-loss distances  the $\PT$-symmetric waveguide feature a ladder of resonances, which resembles the well-known Wannier-Stark ladder for periodic lattices with an additional linear potential.  The number of resonances increases as the gain-to-loss distance grows. Moreover, changing the gain-and-loss amplitude, one can transform a ladder of resonances to a ladder of complex-conjugate eigenvalues --- the feature which is not accessible in Hermitian (self-adjoint) systems.

The transformation from a resonance to an eigenvalue corresponds to a self-dual spectral singularity, i.e. to a zero-width resonance in the continuous spectrum. Transition through a self-dual spectral singularity corresponds to a bifurcation of a complex-conjugate pair of eigenvalues from an interior point of the continuous spectrum. Our second main result consists in the detailed study of spectral singularities that emerge in the system with varying  
gain-to-loss distance. We have demonstrated that by choosing a proper distance between the gain and loss, one can achieve a spectral singularity at any wavenumber (which is not possible in the waveguide with adjacent gain and losses, where spectral singularities occur only at certain isolated values of the wavenumber). We have also demonstrated that the increase of the distance between gain and loss reduces the $\PT$-symmetry breaking threshold and, respectively, allows to achieve  coherent perfect absorption-laser operation (alias, laser-antilaser regime) at a lower value of the gain-and-loss amplitude. Using the fact that   the equation  determining  the  spectral singularities is periodic in the gain-to-loss distance $\ell$, we have demonstrated how one can generate systematically spectral singularities at subsequently increasing distances. In particular, it is possible to generate the branches of spectral singularities that do not have a counterpart in the limit $\ell=0$.
By tuning the gain-and-loss amplitude and a distance, one can reach a situation when two different self-dual spectral singularities  emerge \emph{simultaneously} (but, generically speaking, with different wavenumbers). This  suggests that two pairs of complex-conjugate eigenvalues can emerge from two different interior points of the continuous spectrum simultaneously. Additionally, this observation defies   conjecture~(ii) in Conclusion of \cite{Ahmed18} where it is suggested that a parametrically fixed  $\PT$-symmetric complex  potential has at most one spectral singularity.  Thus, in spite of its apparent simplicity, the  $\PT$-symmetric model with spaced gain and loss elements essentially enriches the previously known phenomenology.

\section*{Acknowledgments}

The results concerning the large distance regime (Lemmata~\ref{lm3.4},~\ref{lm3.6},~\ref{lm3.7} and Section~\ref{sec:ladders}) were financially supported by Russian Science Foundation (project No. 17-11-01004).
The results by D.A.Z. on spectral singularities  were obtained with the support  from  Russian Foundation for Basic Research,  project No. 19-02-00193$\backslash$19.

\end{document}